\documentclass[a4paper,11pt]{article}

\usepackage{graphicx}
\usepackage{epsfig}
\usepackage{natbib}

\usepackage{url}
\usepackage{amsmath,amsfonts,bm,mathrsfs,amssymb,amsthm}
\usepackage{subfig}
\usepackage{multirow}
\usepackage{array}
\usepackage{multicol}
\usepackage{color}
\usepackage{slashbox}
\usepackage[normalem]{ulem}
\usepackage{caption}

\newtheorem{theorem}{Theorem}

\newtheorem{lemma}{Lemma}

\newtheorem{definition}{Definition}

\newtheorem{remark}{Remark}
\newtheorem{example}{\bf Example}

\usepackage{bm}
\usepackage{pgfplots}
\usepackage{tikz}
\usepackage[framemethod=tikz]{mdframed}
\usepackage{lipsum}
\usepackage[symbol]{footmisc}
\usetikzlibrary{arrows,calc,shapes, snakes, intersections}

\usepgfplotslibrary{fillbetween}
\usetikzlibrary{patterns}

\begin{document}

\title{On Deterministic Optimal Mechanisms in a Two-Item Setting for Distributions with Nondecreasing Density}
\author{D.~Thirumulanathan\\ Department of Economic Sciences,\\ Indian Institute of Technology, Kanpur, India}

\maketitle

\begin{abstract}
Consider the problem of designing a revenue-optimal auction mechanism when two heterogeneous items are sold to a single buyer having independent valuations over the items. The distributions of the buyer's valuation for the items are assumed to have densities that are positive, nondecreasing, and continuously differentiable on their support sets $[c_i,c_i+b_i]$ in the positive axis. I prove that the optimal mechanism is deterministic if at least one of the minimum valuations (i.e., either $c_1$ or $c_2$) is sufficiently high. I provide a method to calculate the threshold of $(c_1,c_2)$ beyond which the optimal mechanism is deterministic. I also provide a sufficient condition on the distributions of buyer's valuations for which the individual sale mechanism is optimal.

I show that when $c_1$ is low and $c_2$ is high, it is optimal for the seller to sell item $2$ at the minimum valuation $c_2$, thus effectively reducing the problem to finding the optimal mechanism in the one-dimensional setting only for item $1$. I conjecture with promising preliminary results that this result can be extended to the three-item setting. Specifically, I conjecture that when $c_1$ and $c_2$ are low but $c_3$ is high, it is optimal for the seller to sell item $3$ at the minimum valuation $c_3$, thus effectively reducing the problem to finding the optimal mechanism in the two-dimensional setting for items $1$ and $2$.
\end{abstract}

\section{Introduction}\label{sec:intro}
Consider the problem of designing the revenue-optimal mechanism for selling two heterogeneous items to a single buyer. The buyer's valuation for the items, $z$, is his private information, but the distribution from which he picks his valuation, $f$, is assumed to be common knowledge. The seller aims to design a mechanism that maximizes his expected revenue, where the expectation is taken over the distribution $f$.

The general solution to this problem in the one-item setting was computed by \citet{Myerson81}. He showed that the revenue-optimal mechanism in this setting has a simple structure; it is just a {\em take-it-or-leave-it offer} for a reserve price that depends on $f$. However, finding a general solution in the two-item setting has been a notoriously hard problem. Furthermore, the revenue-optimal mechanism in this setting has shown a variety of structures. For example, the optimal mechanism for certain distributions has infinite menus (\cite{DDT13, DDT17}). Also, the optimal mechanism for the same distribution with different support sets can have many different structures (\cite{Thiru16, Thiru19}).

In this paper, I consider a particular class of distributions and show that the optimal mechanisms are deterministic when the minimum valuation of at least one of the items is high. More specifically, I assume that the distribution of buyer's valuation for item $i$ has a strictly positive, nondecreasing, and continuously differentiable density $f_i$ with a support set $[c_i,c_i+b_i]$. I analyze how the optimal mechanism varies as the support set $[c_1,c_1+b_1]\times[c_2,c_2+b_2]$ varies. The main result, Theorem \ref{thm:individual-sale}, shows that it is optimal for the auctioneer to sell the items individually using Myerson's revenue-optimal mechanism, either when $c_1$ is low and $c_2$ is high, or when $c_2$ is low and $c_1$ is high; and to sell the items as a bundle when both $c_1$ and $c_2$ are high. See Figure \ref{fig:optimal-illust}.

\begin{figure}[h!]
\centering
\begin{tikzpicture}[scale=0.4,font=\small,axis/.style={very thick, ->, >=stealth'}]
\draw [axis,thick,->] (0,-1)--(0,12);
\node [above] at (0,12) {$c_2$};
\draw [axis,thick,->] (-1,0)--(12,0);
\node [right] at (12,0) {$c_1$};
\node at (-0.3,-0.4) {$0$};
\draw [axis,thick,->] (2,3.5)--(2,6.5);
\node [rotate=90] at (1.5,5) {\tiny\bf Individual};
\draw [axis,thick,-] (0,7)--(4,7);
\draw [axis,thick,-] (0,7)--(0,11);
\draw [axis,thick,-] (0,11)--(4,11);
\draw [axis,thick,-] (4,11)--(4,7);
\draw [axis,thick,-] (2,7)--(2,11);
\node at (1,9) {\scriptsize$(0,1)$};
\node at (3,9) {\scriptsize$(1,1)$};
\draw [axis,thick,->] (3.5,2)--(6.5,2);
\node at (5,1.5) {\tiny\bf Individual};
\draw [axis,thick,-] (7,0)--(7,4);
\draw [axis,thick,-] (7,0)--(11,0);
\draw [axis,thick,-] (11,0)--(11,4);
\draw [axis,thick,-] (11,4)--(7,4);
\draw [axis,thick,-] (7,2)--(11,2);
\node at (9,1) {\scriptsize$(1,0)$};
\node at (9,3) {\scriptsize$(1,1)$};
\draw [axis,thick,->] (3.5,3.5)--(6.5,6.5);
\node [rotate=45] at (4.75,5.25) {\tiny\bf Bundle};
\draw [axis,thick,-] (7,7)--(7,11);
\draw [axis,thick,-] (7,11)--(11,11);
\draw [axis,thick,-] (11,11)--(11,7);
\draw [axis,thick,-] (11,7)--(7,7);
\draw [axis,thick,-] (7,10)--(10,7);
\node at (7.9,7.9) {\scriptsize$(0,0)$};
\node at (9.5,9.5) {\scriptsize$(1,1)$};
\end{tikzpicture}
\caption{A heuristic illustration of optimal mechanism for various values of $(c_1,c_2)$. It is optimal to sell the items individually, either when $c_1$ is low and $c_2$ is high, or when $c_2$ is low and $c _1$ is high. It is optimal to sell the items as a bundle when both $c_1$ and $c_2$ are high.}\label{fig:optimal-illust}
\end{figure}
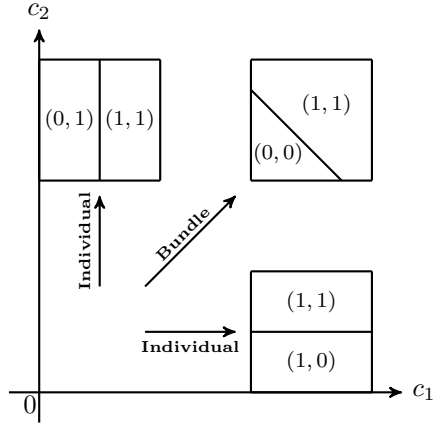

The main result in the paper can be interpreted as follows. Starting from $(c_1,c_2)=(0,0)$, assume that the support set of the distribution goes to infinity either vertically or horizontally or both. I show that the optimal mechanism for all distributions under consideration ends up either as a bundle sale or as an individual sale, based on whether both $c_1$ and $c_2$ are high, or only one of them is high. The optimal mechanism thus ends up being deterministic either way. The interpretation is illustrated in Figure \ref{fig:optimal-illust}.

Furthermore, I show that for every $c_1$, there exists a $c_2^*(c_1)$ such that the optimal mechanism is deterministic for all $c_2\geq c_2^*(c_1)$. I also provide a method to calculate the threshold $c_2^*(c_1)$ for any given $c_1$. A similar assertion holds for every $c_2$. In addition, I also derive sufficient conditions on the densities $f_i$ for which the individual sale mechanism is optimal. The conditions are largely based on the virtual valuation functions and the hazard rate condition and are thus easily verifiable for any $f_i$.

Another interpretation of the main result (Theorem \ref{thm:individual-sale}) is as follows. When $c_1$ is low and $c_2$ is high, the seller finds it optimal to sell the item $2$ at its minimum valuation $c_2$. So the problem effectively reduces to finding the optimal mechanism in the one-dimensional setting only for item $1$. I conjecture that this property can be extended to the three-item setting. Specifically, I consider the example where the valuations for the items are distributed independently according to Unif$[0,1]\times[0,1]\times[c_3,c_3+1]$ with $c_3\geq3$, and show that the optimal mechanism in this example is to sell the item $3$ at its minimum valuation $c_3$, and items $1$ and $2$ according to the optimal mechanism in the two-dimensional setting. Using this example, I conjecture that when $c_1$ and $c_2$ are low but $c_3$ is high, the problem effectively reduces to finding the optimal mechanism in the two-dimensional setting for items $1$ and $2$.

\subsection{Method}
My method begins from the optimal menu theorem (\cite{DDT17}). The theorem provides a necessary and sufficient condition for the optimality of any mechanism with a finite number of constant allocation regions. The necessary and sufficient conditions are based on second-order stochastic dominance of measures defined on a two-dimensional space. These are conditions that are hard to verify in practice. In this paper, I consider the individual sale mechanism as a special case, and use optimal menu theorem to derive sufficient conditions on the support set parameters $(c_1,c_2)$ to certify the optimality of this mechanism.

More specifically, I restrict attention to distributions whose densities are positive, nondecreasing and continuously differentiable on their domains, and find that the individual sale mechanism is optimal either when $c_1$ is low and $c_2$ is high, or when $c_2$ is low and $c_1$ is high. Furthermore, I observe that the problem of finding the optimal mechanism in two-dimensional setting reduces to a problem in the one-dimensional setting. The technical contribution of this paper is in showing the second-order stochastic dominance of measures defined on a two-dimensional space for a wide class of distributions.

I then consider the bundle sale mechanism and calculate the values of $(c_1,c_2)$ for which the mechanism is optimal. The values are calculated in a straightforward manner using a result from \cite{Menicucci15} that provides sufficient conditions for the optimality of the bundle sale mechanism in particular. I show that bundle sale mechanism is optimal when both $c_1$ and $c_2$ are sufficiently high. Furthermore, I also provide a method to calculate the threshold values of $(c_1,c_2)$ from which either the individual sale or the bundle sale mechanism is optimal.

I finally consider the three-dimensional setting in order to check if analogous results can be obtained. I first show using optimal menu theorem that the bundle sale mechanism is optimal when all of $c_1$, $c_2$, and $c_3$ are high. I then consider a specific example when the valuations are uniformly distributed, and show using optimal menu theorem that when $(c_1,c_2)=(0,0)$ but $c_3\geq3$, the three-dimensional problem reduces to a two-dimensional problem. I thus show the second-order stochastic dominance of a measure defined on a {\em three-dimensional space}, but only for a specific example.

\subsection{Prior Work}
Design of multi-dimensional optimal mechanisms has been an area of interest in the literature for over four decades since Myerson's characterization of optimal mechanism in the one-dimensional setting in 1981 (\cite{Myerson81}). See \citet{RS03} for a comprehensive survey on this literature. In this subsection, I restrict attention to works that either derive the exact solutions in the multi-dimensional setting, or the works that derive qualitative results using virtual valuations or the hazard rate condition.

In the two-item setting, \citet{DDT13} and \citet{GK15} have computed the optimal mechanisms explicitly for a large class of distributions. In particular, \citet{DDT13} designed an algorithm to compute the optimal mechanism when the distribution gives rise to a {\em well-formed canonical partition}. However, they considered distributions whose support set is of the form $[0,b_1]\times[0,b_2]$. Optimal mechanism for distributions having a general support set, $[c_1,c_1+b_1]\times[c_2,c_2+b_2]$, has been derived in the literature only for a very few distributions. In this paper, I analyze the variation in the optimal mechanism when the support set of the distribution varies.

Regarding distributions with general support set, \citet{DDT17} prove the optimal menu theorem and provide explicit solutions for some example distributions with support set that are not bordered by the coordinate axes. \citet{Thiru16, Thiru19} restricted attention to uniform distribution, and derived the exact optimal mechanism for all possible support sets $[c_1,c_1+b_1]\times[c_2,c_2+b_2]$. In the $n$-item setting, \citet{GK14} have provided the exact solutions (for $n\leq 6$) when the distribution of buyer's valuation is uniform in $[0,1]^n$. \citet{BNR18} showed that the optimal deterministic mechanism in the two-item setting for any distribution of the buyer's valuations is submodular and monotone, but may be supermodular and non-monotone when the number of items is more than two. \citet{KKR26} considered a unit-demand setting in $n$-items and showed that selling each item at an identical price is optimal when the distribution of buyer's valuation is scale-monotone. \citet{YWJY19} considered a multi-item multi-buyer setting, and showed that for every incentive compatible stochastic mechanism, there exists an equivalent deterministic mechanism that delivers the same social surplus.

Several papers in the literature have obtained qualitative results on optimal mechanism using virtual valuation function and the hazard rate condition. Papers by \citet{Menicucci15}, \citet{Pavlov11}, and \citet{HH15, HH21} show optimality of certain mechanisms using virtual valuations of the buyer. Papers by \citet{Pav10}, \citet{TW17}, \cite{WT14}, and \citet{MV06} consider distributions satisfying the hazard rate condition, and show qualitative results on optimal mechanism in the two-item setting. In my paper, both virtual valuations and the hazard rate condition are used to establish the optimality of deterministic mechanisms.

There is also a vast literature on computing the approximations of the optimal mechanism. See \cite{BGN17}, \cite{BILW14}, \cite{CDW16, CZ17, CMS15, CM16}, \cite{HN12}, \cite{Yao14} for relevant literature on approximate solutions. Some recent literature also deal with computation of optimal mechanisms using deep learning methods. See \cite{Dutting19, Dutting21, Dutting24}, \cite{You23}, for example. My paper, however, focuses on deriving the exact solutions analytically.

The rest of the paper is organized as follows. Section \ref{sec:prelim} describes the two-item one-buyer setting considered in the paper along with the definitions of the terms that will be used in the rest of the paper. The main results regarding the optimal deterministic mechanisms for distributions having nondecreasing densities are stated and proved in Section \ref{sec:main-results}. The extensions to the three-item setting are described in Section \ref{sec:ext}. The paper is concluded in Section \ref{sec:conclusion}.

\section{Preliminaries}\label{sec:prelim}
Consider the problem of designing the revenue-optimal mechanisms in an $n$-item one-buyer setting when the exact valuation of the buyer for the items, $z=(z_1,\ldots,z_n)$, is his private information, but the distribution from which the valuation is picked is common knowledge. I assume that $z_i$, $i=1,2,\ldots,n$, are absolutely continuous random variables that are independent of each other, and that $z_i\sim f_i$, and $(z_1,\ldots,z_n)\sim f(=\times_{i=1}^n f_i)$. The cumulative distribution functions are denoted by $F_i$. The support set of $f_i$ is assumed to be $[c_i,c_i+b_i]$ for some $c_i, b_i\geq 0$. Let $D=\times_{i=1}^n[c_i,c_i+b_i]$.

By revelation principle, we restrict attention to direct mechanisms where the design of an auction mechanism is equivalent to the design of an allocation function $q:D\rightarrow[0,1]^n$ and a payment function $t:D\rightarrow\mathbb{R}_+$. The seller designs the mechanism $(q,t)$ to maximize his expected revenue. Observe that the mechanism is not restricted to be deterministic; the buyer can be allocated item $i$ with some probability $q_i$ based on his valuations. Given that the exact valuation of the buyer is not known to the seller, he designs $(q,t)$ based on the report of the buyer $\hat{z}=(\hat{z}_1,\ldots,\hat{z}_n)$. The utility of the buyer, $u:D^2\rightarrow\mathbb{R}$, is assumed to be of the form
$$
  u(\hat{z}_1,\ldots,\hat{z}_n;z_1,\ldots,z_n)=z_1q_1(\hat{z})+\cdots+z_nq_n(\hat{z})-t(\hat{z}).
$$

The problem of optimal mechanism design is where the seller wishes to design $(q,t)$ so that his expected revenue is maximized, subject to two constraints: {\em incentive compatibility} and {\em individual rationality}. The constraints are defined as follows.
\begin{definition}
A mechanism $(q,t)$ is {\em Incentive Compatible} (IC) if $u(\hat{z};z)\leq u(z;z)$ for all $\hat{z},z\in D$. In other words, the buyer cannot improve his utility by misreporting his valuation.
\end{definition}
\begin{definition}
A mechanism $(q,t)$ is {\em Individually Rational} (IR) if $u(z;z)\geq 0$ for all $z\in D$. In other words, the buyer gets a nonnegative utility when he reports his true valuation.
\end{definition}

In this paper, I show that the optimal mechanism $(q,t)$ is deterministic in the two-item setting when (i) the distribution of the buyer's valuations has a density that is positive, nondecreasing and continuously differentiable, and (ii) the support sets of the densities satisfy certain conditions. The proof of optimality crucially uses the {\em optimal menu theorem} from \cite{DDT17}. I first introduce the related measure-theoretic preliminaries before stating the theorem.

Define a signed measure $\bar{\mu}$ on $D$. The density of $\bar{\mu}$ has the following three components:
\begin{align}
  &\mbox{($n$-dimensional density) }\mu(z)=-(n+1)f(z)-z\cdot\nabla f(z),\,\forall z\in D, \label{eqn:mu} \\
  &\mbox{($(n-1)$-dimensional density) }\mu_s(z)=(z\cdot\hat{n}(z))f(z),\,\forall z\in\partial D, \label{eqn:mu-s} \\
  &\mbox{(Point density) }\mu_p(c_1,c_2)=1, \label{eqn:mu-p}
\end{align}
where $\partial D$ refers to the boundary of the rectangle $D$, and $\hat{n}(z)$ refers to the outer unit normal vector at the boundary point $z\in\partial D$. I now recall the definitions of stochastic orders used in the paper.
\begin{definition}
Consider a signed measure $\alpha$ defined on $D$. We say that $\alpha$ dominates zero in the first-order ($\alpha\succeq_10$) if $\int_Dh\,d\alpha\geq0$ holds for all increasing functions $h:D\rightarrow\mathbb{R}$.
\end{definition}
\begin{definition}
Consider a signed measure $\alpha$ defined on $D$. We say that $\alpha$ dominates zero in $cvx(\vec{v})$-order ($\alpha\succeq_{cvx(\vec{v})}0)$ for some $\vec{v}=(v_1,\ldots,v_n)\in\{-1,0,1\}^n$ if $\int_Dh\,d\alpha\geq0$ for all $h:D\rightarrow\mathbb{R}$ that is (i) convex, (ii) nondecreasing in coordinate $i$ if $v_i=1$, and (iii) nonincreasing in coordinate $i$ if $v_i=-1$.
\end{definition}

Define a menu of a mechanism as the set of allocation choices available in the mechanism. A formal definition of a menu $M$ is as follows.
$$
  M=\{(\hat{q},\hat{t}):\exists z\in D, (\hat{q},\hat{t})=(q(z),t(z))\}.
$$

Define constant allocation region $R_{\hat{q},\hat{t}}=\{z:(q(z),t(z))=(\hat{q},\hat{t})\}$. It represents the set of valuations for which the buyer has a constant allocation $\hat{q}$ and a constant payment $\hat{t}$. Let $\bar{\mu}|_{R_{\hat{q},\hat{t}}}$ refer to the $\bar{\mu}$-measure restricted to a particular constant allocation region $R_{\hat{q},\hat{t}}$. The optimal menu theorem is as follows.
\begin{theorem} \cite[Thm.~3]{DDT17}\label{thm:opt_menu}
Consider a mechanism with $|M|<\infty$. For a constant allocation region $R_{\hat{q},\hat{t}}$, define $\vec{v}(\hat{q},\hat{t})=(v_1(\hat{q},\hat{t}),\ldots,v_n(\hat{q},\hat{t}))$ as
$$
  v_i(\hat{q},\hat{t})=\begin{cases}1&\mbox{if }\hat{q}_i=0\\-1&\mbox{if }\hat{q}_i=1\\0&\mbox{otherwise.}\end{cases}
$$
Then, the mechanism is optimal if and only if $\bar{\mu}|_{R_{\hat{q},\hat{t}}}\preceq_{cvx(\vec{v}(\hat{q},\hat{t}))}0$ for every constant allocation region $R_{\hat{q},\hat{t}}$.
\end{theorem}

I use this theorem to show the optimality of individual sale mechanisms in the two-item setting considered in this paper. The individual sale mechanism involves selling each of the two items separately using the optimal mechanism in the one-item setting. Defining the {\em modified virtual valuation function} $\tilde{\phi}_i:[c_i,c_i+b_i]\rightarrow\mathbb{R}$ as $\tilde{\phi}_i(z_i)=z_if_i(z_i)-(1-F_i(z_i))$, we have the following result from \cite{Myerson81} regarding the one-item setting.\footnote{\citet{Myerson81} defines $\phi_i:[c_i,c_i+b_i]\rightarrow\mathbb{R}$ as $\phi_i(z_i)=z_i-\frac{1-F_i(z_i)}{f_i(z_i)}$ and calls it the virtual valuation function. We consider a modification of this function $\tilde{\phi}_i$ as the {\em modified virtual valuation} function. The modified function is more relevant in our analysis and is used throughout the paper.}
\begin{theorem}\label{thm:myerson}
The optimal mechanism in the one-item setting is a take-it-or-leave-it offer for a reserve price $p=\arg\max_{z_1}(z_1(1-F_1(z_1))$. If $\tilde{\phi}_1(\cdot)$ is strictly increasing and $\tilde{\phi}_1^{-1}(0)$ exists, then $p=\arg\max_{z_1}(z_1(1-F_1(z_1))=\tilde{\phi}_1^{-1}(0)$.
\end{theorem}

I use the following result from \cite{Menicucci15} to show the optimality of bundle sale mechanism in the setting considered in this paper.
\begin{theorem}\cite[Prop.~3]{Menicucci15}\label{thm:menicucci}
Consider that the distributions $f_1$ and $f_2$ are positive, differentiable, have bounded derivatives, and are such that (i) $\mu(z)\leq0$ for all $z\in D$, and (ii) $\tilde{\phi}_i(z_i)\geq 0$ for all $z_i\in[c_i,c_i+b_i]$, $i=1,2$. Then, the optimal mechanism is to sell the items as a bundle.
\end{theorem}

\section{Main Results}\label{sec:main-results}
I begin with the setting where two heterogeneous items are sold to one buyer. I consider that the distributions of the valuations of the buyer satisfy some sufficient conditions, and show that the optimal mechanism is deterministic when the minimum valuation of at least one of the items is sufficiently high.
\subsection{Optimal Mechanisms in the Two-Item Setting}\label{sec:two-item}
Consider two distributions whose densities are given by $g_i:[0,b_i]\rightarrow\mathbb{R}_+$, $i=1,2$, each of which is strictly positive, nondecreasing and continuously differentiable in its domain. The distribution of buyer's valuation for item $i$ is given by
$$
  f_i(z_i)=\begin{cases}g_i(z_i-c_i)&\mbox{if }z_i\in[c_i,c_i+b_i],\\0&\mbox{else},\end{cases}
$$
for some $c_i\geq0$. A simple example of such a setting is when $z_i\sim\mbox{Unif}[c_i,c_i+b_i]$, where we have $g_i(z_i)=\frac{1}{b_i}$ when $z_i\in[0,b_i]$ and $0$ otherwise. Observe that $g_i$'s are strictly positive, nondecreasing and continuously differentiable in their domains.

I analyze how the optimal mechanism varies with $(c_1,c_2)$. In other words, I fix a base distribution $g(\cdot)$, and analyze how the optimal mechanism varies when the base distribution remains the same but the support set of the distribution varies. I begin by proving certain properties of the $n$-dimensional density $\mu$ (defined in \eqref{eqn:mu}) and the modified virtual valuation function $\tilde{\phi}_i(\cdot)$ for all distributions under consideration.
\begin{lemma}\label{lem:mu-phi-i}
Consider the density functions $g_1$ and $g_2$ to be strictly positive, nondecreasing, and continuously differentiable. Let $f_i:[c_i,c_i+b_i]\rightarrow\mathbb{R}_+$ be $f_i(z_i)=g_i(z_i-c_i)$. Then,
\begin{enumerate}
 \item[(a)] For every $c_1,c_2\geq0$, we have $\mu(z)<0,\,\forall z\in D$.
 \item[(b)] $\tilde{\phi}_i(\cdot)$ is a strictly increasing function for every $c_i\geq0$.
 \item[(c)] For every $c_i<\frac{1}{g_i(0)}$, $\tilde{\phi}_i^{-1}(0)$ exists and is unique. Furthermore, $\tilde{\phi}_i^{-1}(0)\in(c_i,c_i+b_i)$.
\end{enumerate}
\end{lemma}
{\bf Proof:} See Appendix A. \qed

\begin{remark}
The condition $\mu(z)<0$ for all $z\in D$, as shown in Lemma \ref{lem:mu-phi-i}(a), is a standard regularity condition used in the optimal multi-dimensional mechanism design literature, but has been denoted by different terms. For example, it is denoted as the {\em hazard rate condition} in \cite{Pav10}, as {\em Mhr1 condition} in \cite{MV06} and as a condition based on {\em power rate of a distribution} in \cite{WT14} and \cite{TW17}. Lemma \ref{lem:mu-phi-i}(a) shows that the condition holds for all $c_i\geq0$ when $g_i$'s are nondecreasing.
\end{remark}

The results shown in Lemma \ref{lem:mu-phi-i} hold true for uniform distributions given that $g_i$ satisfies all the conditions given in the lemma. Note that for uniform distributions, we have
\begin{enumerate}
 \item[(a)] $\mu(z)=-z\cdot\nabla f(z)-3f(z)=-3,\,\forall z\in D$ for every $c_1,c_2\geq0$,
 \item[(b)] $F_i(z_i)=\frac{z_i-c_i}{b_i}$ when $z_i\in[c_i,c_i+b_i]$, and so we have $\tilde{\phi}_i(z_i)=\frac{z_i}{b_i}-1+\frac{z_i-c_i}{b_i}=\frac{2z_i-(c_i+b_i)}{b_i}$. The function is strictly increasing for every $c_i\geq0$.
 \item[(c)] $\tilde{\phi}_i^{-1}(0)=\frac{c_i+b_i}{2}\in(c_i,c_i+b_i)$ whenever $c_i<\frac{1}{g_i(0)}=b_i$.
\end{enumerate}

I now proceed to characterize the threshold for which the minimum valuations are termed ``sufficiently high'' for the optimal mechanism to be deterministic. I define three functions $c_2^{th_k}:\left[0,\frac{1}{g_1(0)}\right)\rightarrow\mathbb{R}_+$, $k=1,2,3$, as follows:
\begin{itemize}
 \item $c_2^{th_1}(c_1)=\frac{1}{g_2(0)F_1(\tilde{\phi}_1^{-1}(0))}$. The function is well-defined since (i) $g_2(0)>0$ and (ii) $\tilde{\phi}_1^{-1}(0)>c_1$ and $f_1>0$ imply that $F_1(\tilde{\phi}_1^{-1}(0))>0$.
 \begin{itemize}
  \item For uniform distributions, we have $F_1(\tilde{\phi}_1^{-1}(0))=\frac{b_1-c_1}{2b_1}$. So we have $c_2^{th_1}(c_1)=2b_2\frac{b_1}{b_1-c_1}$.
 \end{itemize}
 \item $c_2^{th_2}(c_1)=\frac{1}{g_2(0)}\left(2+\max_{z_1\in[c_1,\tilde{\phi}_1^{-1}(0)]}\frac{z_1f_1'(z_1)}{f_1(z_1)}\right)$. The function is well-defined since $f_i(z_i)>0$, $f_i'(z_i)$ is continuous, and $[c_1,\tilde{\phi}_1^{-1}(0)]$ is a compact set.
 \begin{itemize}
  \item For uniform distributions, we have $\frac{z_1f_1'(z_1)}{f_1(z_1)}=0$ for all $z_1$. So we have $c_2^{th_2}(c_1)=2b_2$.
 \end{itemize}
 \item Define $t:\left[\frac{1}{g_2(0)},\infty\right)\rightarrow(0,b_1]$, $t(c_2)=\{t:c_2g_2(0)F_1(c_1+t)=1\}$. Now I define
$$
  c_2^{th_3}(c_1)=\left\{c_2:\frac{\int_{c_1}^{c_1+t(c_2)}z_1f_1(z_1)\,dz_1}{F_1(c_1+t(c_2))}=\max_{z_1\in[c_1,c_1+b_1]}(z_1(1-F_1(z_1)))\right\}.
$$
\end{itemize}

The function $c_2^{th_3}(c_1)$ can be interpreted as follows. Consider $\hat{f}_1$ as the density function formed by the truncation of $f_1$ in the interval $[c_1,c_1+t]$. In other words, let $\hat{f}_1:[c_1,c_1+t]\rightarrow\mathbb{R}_+$, $\hat{f}_1(z_1)=\frac{f_1(z_1)}{F_1(c_1+t)}$. Then the equation
$$
  \int_{c_1}^{c_1+t}z_1\left(\frac{f_1(z_1)}{F_1(c_1+t)}\right)\,dz_1=\max_{z_1\in[c_1,c_1+b_1]}z_1(1-F_1(z_1))
$$
finds the point of truncation $t$, at which the expectation of $z_1$ (when $z_1\sim\hat{f}_1$) equals the maximum expected revenue generated by selling item $1$.

The following lemma shows that  $c_2^{th_3}(\cdot)$ is well-defined.
\begin{lemma}\label{prop:c2-th-3}
Let $t:\left[\frac{1}{g_2(0)},\infty\right)\rightarrow(0,b_1]$, $t(c_2)=\{t:c_2g_2(0)F_1(c_1+t)=1\}$. Then the term $\frac{\int_{c_1}^{c_1+t(c_2)}z_1f_1(z_1)\,dz_1}{F_1(c_1+t(c_2))}$ decreases with $c_2$. Furthermore, for every $c_1<\frac{1}{g_1(0)}$, there exists a unique $c_2\in\left[\frac{1}{g_2(0)},\infty\right)$ such that
$$
  \frac{\int_{c_1}^{c_1+t(c_2)}z_1f_1(z_1)\,dz_1}{F_1(c_1+t(c_2))}=\max_{z_1\in[c_1,c_1+b_1]}(z_1(1-F_1(z_1))).
$$
\end{lemma}
{\bf Proof:} See Appendix A.\qed

For uniform distributions, we have $t(c_2)=\{t:\frac{c_2}{b_2}\cdot\frac{t}{b_1}=1\}=\frac{b_1b_2}{c_2}$. Also, $F_1(c_1+t)=\frac{b_2}{c_2}$. Furthermore, $z_1(1-F_1(z_1))$ is maximized at $z_1^*=\tilde{\phi}_1^{-1}(0)=\frac{c_1+b_1}{2}$. So we need to solve for $c_2$ for which $\frac{c_2}{b_2}\int_{c_1}^{c_1+\frac{b_1b_2}{c_2}}\frac{z_1}{b_1}\,dz_1=\frac{(c_1+b_1)^2}{4b_1}$.
\begin{multline*}
  \frac{c_2}{2b_1b_2}\left(\left(c_1+\frac{b_1b_2}{c_2}\right)^2-c_1^2\right)=\frac{(c_1+b_1)^2}{4b_1} \\
  \Rightarrow\frac{c_2}{2b_1b_2}\left(\left(\frac{b_1b_2}{c_2}\right)^2+\frac{2c_1b_1b_2}{c_2}\right)=\frac{(c_1+b_1)^2}{4b_1} \\
  \Rightarrow\frac{b_1b_2}{2c_2}=\frac{(c_1+b_1)^2}{4b_1}-c_1\Rightarrow c_2^{th_3}(c_1)=2b_2\frac{b_1^2}{(b_1-c_1)^2}.
\end{multline*} 

Observe that the term $\frac{\int_{c_1}^{c_1+t(c_2)}z_1f_1(z_1)\,dz_1}{F_1(c_1+t(c_2))}=\frac{b_1b_2}{2c_2}-c_1$ decreases with $c_2$, and that there exists a unique $c_2^{th_3}(c_1)\geq b_2$ for every $c_1<b_1$. The uniform distribution thus satisfies Lemma \ref{prop:c2-th-3}.

I further define the threshold $c_2^*:\left[0,\frac{1}{g_1(0)}\right)\rightarrow\mathbb{R}_+$, as $c_2^*(c_1)=\max_{k=1,2,3}(c_2^{th_k}(c_1))$. It is easy to see that $c_2^*(c_1)$ for uniform distributions is
$$
  c_2^*(c_1)=2b_2\max\left(\frac{b_1}{b_1-c_1},1,\frac{b_1^2}{(b_1-c_1)^2}\right)=2b_2\frac{b_1^2}{(b_1-c_1)^2}.
$$

The functions $c_1^{th_k}(c_2)$ for $k=1,2,3$ can be defined similarly, and the threshold function $c_1^*:\left[0,\frac{1}{g_2(0)}\right)\rightarrow\mathbb{R}_+$ can also be defined in a similar manner as $c_1^*(c_2)=\max_{k=1,2,3}(c_1^{th_k}(c_2))$.

\begin{figure}
\centering
\begin{tabular}{cc}
\subfloat[]{\label{fig:e1}\begin{tikzpicture}[scale=0.3,font=\footnotesize,axis/.style={very thick, -}]
\node at (0,-1) {\tiny$(c_1,c_2)$};
\draw [axis,thick,-] (0,0)--(12,0);
\node at (5,-1.5) {\tiny$\left(\tilde{\phi}_1^{-1}(0),c_2\right)$};
\node at (11,-1) {\tiny$(c_1+b_1,c_2)$};
\draw [axis,thick,-] (0,0)--(0,12);
\node [above] at (1,12) {\tiny$(c_1,c_2+b_2)$};
\draw [axis,thick,-] (0,12)--(12,12);
\draw [axis,thick,-] (12,0)--(12,12);
\draw [axis,thick,-] (6,0)--(6,12);
\node at (3,6) {$(0,1)$};
\node at (9,6) {$(1,1)$};
\end{tikzpicture}}&
\subfloat[]{\label{fig:h1}\begin{tikzpicture}[scale=0.3,font=\footnotesize,axis/.style={very thick, -}]
\node at (1.5,-1) {\tiny$(c_1,c_2)$};
\draw [axis,thick,-] (0,0)--(12,0);
\node at (11,-1) {\tiny$(c_1+b_1,c_2)$};
\draw [axis,thick,-] (0,0)--(0,12);
\node [above] at (3,12) {\tiny$(c_1,c_2+b_2)$};
\draw [axis,thick,-] (0,12)--(12,12);
\draw [axis,thick,-] (12,0)--(12,12);
\draw [axis,thick,-] (0,6)--(12,6);
\node [rotate=90] at (13.5,6) {\tiny$\left(c_1+b_1,\tilde{\phi}_2^{-1}(0)\right)$};
\node at (6,3) {$(1,0)$};
\node at (6,9) {$(1,1)$};
\end{tikzpicture}}
\end{tabular}
\caption{Structure of the optimal mechanism when (a) $c_1<\frac{1}{g_1(0)}$ and $c_2$ is high, (b) $c_2<\frac{1}{g_2(0)}$ and $c_1$ is high.}\label{fig:structure-low-high}
\end{figure}
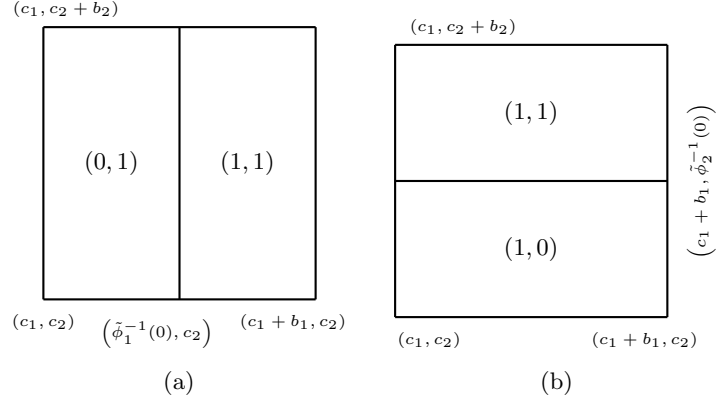

I am now ready to state the main result of the paper.
\begin{theorem}\label{thm:individual-sale}
Consider two distributions whose density functions, $g_i:[0,b_i]\rightarrow\mathbb{R}_+$, $i=1,2$, are strictly positive, nondecreasing and continuously differentiable. Let the support set of the distribution of the buyer's valuation be $[c_1,c_1+b_1]\times[c_2,c_2+b_2]$, and the density of the distribution be given by $f(z_1,z_2)=f_1(z_1)f_2(z_2)=g_1(z_1-c_1)g_2(z_2-c_2)$. Then, the following holds:
\begin{enumerate}
 \item[(a)] If $c_1<\frac{1}{g_1(0)}$ and $c_2\geq c_2^*(c_1)$, then the optimal mechanism is to sell the items individually at the prices $\tilde{\phi}_1^{-1}(0)$ and $c_2$ respectively. The structure of the optimal mechanism is as in Figure \ref{fig:e1}.
 \item[(b)] If $c_2<\frac{1}{g_2(0)}$ and $c_1\geq c_1^*(c_2)$, then the optimal mechanism is to sell the items individually at the prices $c_1$ and $\tilde{\phi}_2^{-1}(0)$ respectively. The structure of the optimal mechanism is as in Figure \ref{fig:h1}.
 \item[(c)] If $c_1\geq\frac{1}{g_1(0)}$ and $c_2\geq\frac{1}{g_2(0)}$, then the optimal mechanism is to sell the items as a bundle. The structure of the optimal mechanism is as in Figure \ref{fig:c1}.
\end{enumerate}
\end{theorem}
The theorem asserts that for each $c_1$, there exists a threshold value of $c_2$ beyond which the optimal mechanism is deterministic. Furthermore, the optimal mechanism is to sell the items individually when $c_1<\frac{1}{g_1(0)}$ (and $c_2$ is beyond the corresponding threshold), and to sell the items as a bundle when $c_1\geq\frac{1}{g_1(0)}$ (and $c_2$ is beyond the corresponding threshold). A similar assertion holds for $c_2$.

It is important to note that the theorem does not assert the tightness of the thresholds $c_2^*(c_1)$ and $c_1^*(c_2)$. I provide an example of densities $f_1$ and $f_2$ later in Section \ref{sec:discussion} where $c_1=0$ and $c_2<c_2^*(0)$ but the optimal mechanism still is to sell the items individually at the prices $\tilde{\phi}_1^{-1}(0)$ and $c_2$ respectively.

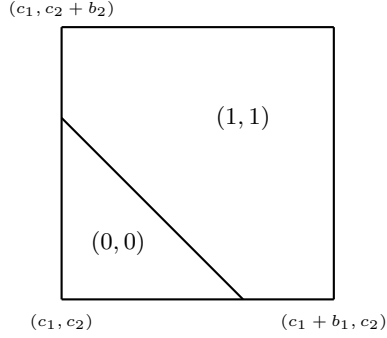
\begin{figure}
\centering
\begin{tikzpicture}[scale=0.3,font=\footnotesize,axis/.style={very thick, -}]
\node at (0,-1) {\tiny$(c_1,c_2)$};
\draw [axis,thick,-] (0,0)--(12,0);
\node at (12,-1) {\tiny$(c_1+b_1,c_2)$};
\draw [axis,thick,-] (0,0)--(0,12);
\node [above] at (0,12) {\tiny$(c_1,c_2+b_2)$};
\draw [axis,thick,-] (0,12)--(12,12);
\draw [axis,thick,-] (12,0)--(12,12);
\draw [axis,thick,-] (0,8)--(8,0);
\node at (2.5,2.5) {$(0,0)$};
\node at (8,8) {$(1,1)$};
\end{tikzpicture}
\caption{Structure of the optimal mechanism when $c_1\geq\frac{1}{g_1(0)}$ and $c_2\geq\frac{1}{g_2(0)}$.}\label{fig:c1}
\end{figure}

Part (a) of Theorem \ref{thm:individual-sale} indicates in the asymmetric setting where the minimum valuation of item $1$ is low and the minimum valuation of item $2$ is high, the optimal mechanism reduces to finding the optimal mechanism in the one-item setting, with item $2$ being sold at the minimum valuation. Recall from Theorem \ref{thm:myerson} that in the one-item setting, it is optimal to sell the item at the price $\tilde{\phi}_1^{-1}(0)$ if it exists and if $\tilde{\phi}_1(\cdot)$ is strictly increasing. The seller thus extracts maximum surplus from item $2$ by selling it at the minimum valuation, but conducts an optimal auction for item $1$, selling it at the price $\tilde{\phi}_1^{-1}(0)$. A symmetric argument holds when the minimum valuation of item $2$ is low and the minimum valuation of item $1$ is high.

When we consider uniform distribution, the theorem says that it is optimal for the seller to sell the items
\begin{enumerate}
 \item[(a)] at the price $\frac{c_1+b_1}{2}$ and $c_2$ respectively, when $\{c_1<b_1,c_2\geq2b_2\frac{b_1^2}{(b_1-c_1)^2}\}$,
 \item[(b)] at the price $c_1$ and $\frac{c_2+b_2}{2}$ respectively, when $\{c_1\geq2b_1\frac{b_2^2}{(b_2-c_2)^2},c_2<b_2\}$,
 \item[(c)] as a bundle when $\{c_1\geq b_1,c_2\geq b_2\}$.
\end{enumerate}
Observe that the result is exactly the same as stated in Theorems 5, 7, and 8 in \citet{Thiru19}.

I now proceed to prove Theorem \ref{thm:individual-sale}. The statements of part (a) and part (b) are symmetric, and thus it suffices to prove one of the parts. Without loss of generality, I prove the statement in part (a). The seller posts the prices of the items as $\tilde{\phi}_1^{-1}(0)$ and $c_2$ respectively. The buyer then gets only item $2$ when $z_1<\tilde{\phi}_1^{-1}(0)$, and gets both the items when $z_1\geq\tilde{\phi}_1^{-1}(0)$. He gets item $2$ anyways because the price of the item equals his minimum valuation. So the mechanism has only two constant-allocation regions:
 \begin{itemize}
  \item $\left[\tilde{\phi}_1^{-1}(0),c_1+b_1\right]\times[c_2,c_2+b_2]$ where $q=(1,1)$, and
  \item $\left[c_1,\tilde{\phi}_1^{-1}(0)\right)\times[c_2,c_2+b_2]$ where $q=(0,1)$.
 \end{itemize}

I call the set of valuations at which the buyer gets the bundle of items as $W$, and the set of valuations where the buyer gets only item $2$ as $A$. I then use the following steps to prove part (a) of the theorem.
\begin{itemize}
 \item According to Theorem \ref{thm:opt_menu}, the optimal mechanism has the following allocation function
$$
  q(z_1,z_2)=\begin{cases}(1,1)&\mbox{if }z\in W,\\(0,1)&\mbox{if }z\in A,\end{cases}
$$
if (i) $\bar{\mu}|_W\preceq_{cvx(\overrightarrow{-1,-1})}0$, and (ii) $\bar{\mu}|_A\preceq_{cvx(\overrightarrow{1,-1})}0$. So I derive the conditions on $c_i$ and $f_i(\cdot)$ that satisfies both the conditions. See Lemmas \ref{prop:W-region} and \ref{prop:A-region}.
 \item I then consider the distributions (i) whose density is strictly positive, nondecreasing and continuously differentiable, and (ii) whose support set is such that $c_1<\frac{1}{f_1(c_1)}$ and $c_2\geq c_2^*(c_1)$. I show that the conditions in Lemmas \ref{prop:W-region} and \ref{prop:A-region} are satisfied for any such distribution.
\end{itemize}

To prove part (c) of the theorem, I consider the distributions whose density (i) is strictly positive, nondecreasing and continuously differentiable, and (ii) is such that $c_i\geq\frac{1}{f_i(c_i)}$, $i=1,2$, and then refer to Theorem \ref{thm:menicucci} to show that for any such distribution, the optimal mechanism is to sell the items as a bundle.

I now derive the conditions on $c_i$ and $f_i(\cdot)$ so that (i) $\bar{\mu}|_W\preceq_{cvx(\overrightarrow{-1,-1})}0$ and (ii) $\bar{\mu}|_A\preceq_{cvx(\overrightarrow{1,-1})}0$ hold.
\begin{lemma}\label{prop:W-region}
Consider that $f_1$ and $f_2$ are positive and continuously differentiable in their domains. Consider that in addition they satisfy
\begin{itemize}
 \item $\mu(z)\leq0$ for all $z\in D$,
 \item $\tilde{\phi}_1(\cdot)$ strictly increasing, $\tilde{\phi}_1(c_1)<0$.
\end{itemize}
Let the set $W$ be given by $W:=\left[\tilde{\phi}_1^{-1}(0),c_1+b_1\right]\times[c_2,c_2+b_2]$. Then, $\bar{\mu}|_W\succeq_10$, and hence $\bar{\mu}|_W\preceq_{cvx(\overrightarrow{-1,-1})}0$.
\end{lemma}

\begin{lemma}\label{prop:A-region}
Consider that $f_1$ and $f_2$ are positive and continuously differentiable in their domains. Consider that in addition they satisfy
\begin{itemize}
 \item $\mu(z)\leq0$ for all $z\in D$,
 \item $\tilde{\phi}_1(\cdot)$ strictly increasing, $\tilde{\phi}_1(c_1)<0$.
 \item $c_2f_2(c_2)f_1(z_1)\geq z_1f_1'(z_1)+2f_1(z_1),\,\forall z_1\in[c_1,\tilde{\phi}_1^{-1}(0)]$.
 \item There exists a $t\leq\tilde{\phi}_1^{-1}(0)-c_1$ such that (a) $c_2f_2(c_2)F_1(c_1+t)=1$ and (b) $c_2f_2(c_2)\left(\int_{c_1}^{c_1+t}z_1f_1(z_1)\,dz_1\right)\leq(\tilde{\phi}_1^{-1}(0))^2f_1(\tilde{\phi}_1^{-1}(0))$ hold.
\end{itemize}
Define $A=\left[c_1,\tilde{\phi}_1^{-1}(0)\right)\times[c_2,c_2+b_2]$. Then $\bar{\mu}|_A\preceq_{cvx\overrightarrow{(1,-1)}}0$.
\end{lemma}
{\bf Proof:} See Appendix A.

We are now ready to prove Theorem \ref{thm:individual-sale}.

{\bf Proof of Theorem \ref{thm:individual-sale}:} The proof uses Theorem \ref{thm:opt_menu}, Lemma \ref{prop:W-region} and Lemma \ref{prop:A-region}, and is relegated to Appendix A.

The following result provides a characterization of the values of the threshold functions $c_2^*(c_1)$ and $c_1^*(c_2)$.

\begin{theorem}\label{thm:threshold}
\begin{enumerate}
 \item[(a)] The threshold $c_2^*(c_1)>\frac{1}{g_2(0)}$ for all $c_1<\frac{1}{g_1(0)}$. Furthermore, $c_2^*(c_1)\rightarrow\infty$ as $c_1\rightarrow\frac{1}{g_1(0)}$.
 \item[(b)] The threshold $c_1^*(c_2)>\frac{1}{g_1(0)}$ for all $c_2<\frac{1}{g_2(0)}$. Furthermore, $c_1^*(c_2)\rightarrow\infty$ as $c_2\rightarrow\frac{1}{g_2(0)}$.
\end{enumerate}
\end{theorem}
{\bf Proof:} See Appendix A.

The values of $(c_1,c_2)$ for which the optimal mechanisms are deterministic is illustrated in Figure \ref{fig:deterministic-illust}. Observe that the threshold $c_2^*(c_1)>\frac{1}{g_2(0)}$ for all $c_1<\frac{1}{g_1(0)}$, and tends to infinity as $c_1$ approaches $\frac{1}{g_1(0)}$. So $c_2^*(c_1)$ is asymptotic to the line $c_1=\frac{1}{g_1(0)}$. Similar observation holds for the threshold $c_1^*(c_2)$.\footnote{The curve $c_2^*(c_1)$ is shown to be increasing with $c_1$ in Figure \ref{fig:deterministic-illust} for the purpose of illustration only. Whether $c_2^*$ is an increasing function of $c_1$ or not is still an open problem.}

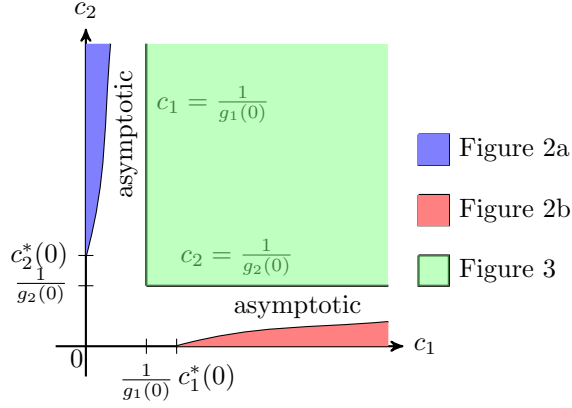
\begin{figure}[h!]
\centering
\begin{tikzpicture}[scale=0.4,font=\small,axis/.style={very thick, ->, >=stealth'}]
\draw [axis,thick,->] (0,-1)--(0,10.5);
\node [above] at (0,10.5) {$c_2$};
\draw [axis,thick,->] (-1,0)--(10.5,0);
\node [right] at (10.5,0) {$c_1$};
\node at (-0.3,-0.4) {$0$};
\draw [thin,-] (-0.25,2) -- (0.25,2);
\node [left] at (-0.25,2) {$\frac{1}{g_2(0)}$};
\draw [thin,-] (-0.25,3) -- (0.25,3);
\node [left] at (-0.25,3) {$c_2^*(0)$};
\draw [thin,-] (2,-0.25) -- (2,0.25);
\node [below] at (2,-0.25) {$\frac{1}{g_1(0)}$};
\draw [thin,-] (3,-0.25) -- (3,0.25);
\node [below] at (4,-0.25) {$c_1^*(0)$};
\draw [thick,-] (2,2)--(2,10);
\node [right] at (2,8) {$c_1=\frac{1}{g_1(0)}$};
\draw [thick,-] (2,2)--(10,2);
\node [above] at (5,2) {$c_2=\frac{1}{g_2(0)}$};
\draw [thick,-] (3,0) to[out=15,in=-175] (10,0.8);
\node [right] at (0.6,7) {\rotatebox{90}{asymptotic}};
\draw [thick,-] (0,3) to[out=75,in=-95] (0.8,10);
\node [above] at (7,0.6) {asymptotic};

\draw [thick,-] (11,7)--(12,7)--(12,6)--(11,6)--(11,7);
\path[fill=blue!50!] (11,7)--(12,7)--(12,6)--(11,6)--(11,7);
\node [right] at (12,6.5) {Figure \ref{fig:e1}};
\draw [thick,-] (11,5)--(12,5)--(12,4)--(11,4)--(11,5);
\path[fill=red!50!] (11,5)--(12,5)--(12,4)--(11,4)--(11,5);
\node [right] at (12,4.5) {Figure \ref{fig:h1}};
\draw [thick,-] (11,3)--(12,3)--(12,2)--(11,2)--(11,3);
\path[fill=green!50,opacity=.5] (11,3)--(12,3)--(12,2)--(11,2)--(11,3);
\node [right] at (12,2.5) {Figure \ref{fig:c1}};

\path[fill=blue!50!] (0,3) to[out=75,in=-95] (0.8,10)--(0,10)--(0,4);
\path[fill=red!50!] (3,0) to[out=15,in=-175] (10,0.8)--(10,0)--(4,0);
\path[fill=green!50,opacity=.5] (2,10)--(2,2)--(10,2)--(10,10)--(2,10);
\end{tikzpicture}
\caption{An illustration of the values of $(c_1,c_2)$ for which the optimal mechanisms are known to be deterministic.}\label{fig:deterministic-illust}
\end{figure}

\subsection{Examples}\label{sec:examples}
I now look at two examples of distributions whose densities $f_1$ and $f_2$ are positive, nondecreasing and continuously differentiable.
\begin{example}\label{eg:unif}
Consider the buyer's valuation to be uniformly distributed in any arbitrary rectangle $D=[c_1,c_1+b_1]\times[c_2,c_2+b_2]$. So we have $g_i(z_i)=\frac{1}{b_i}$ for all $z_i\in[0,b_i]$, and $f_i(z_i)=g_i(z_i-c_i)=\frac{1}{b_i}$ for all $z_i\in[c_i,c_i+b_i]$. The threshold function $c_2^*(c_1)$ was derived in the beginning of Section \ref{sec:two-item} as $c_2^*(c_1)=2b_2\frac{b_1^2}{(b_1-c_1)^2}$. The derivation of $c_1^*(c_2)=2b_1\frac{b_2^2}{(b_2-c_2)^2}$ is similar. Now according to Theorem \ref{thm:individual-sale}, it is optimal for the seller to sell the items
\begin{itemize}
 \item at the price $\frac{c_1+b_1}{2}$ and $c_2$ respectively, when $\{c_1<b_1,c_2\geq2b_2\frac{b_1^2}{(b_1-c_1)^2}\}$,
 \item at the price $c_1$ and $\frac{c_2+b_2}{2}$ respectively, when $\{c_1\geq2b_1\frac{b_2^2}{(b_2-c_2)^2},c_2<b_2\}$,
 \item as a bundle when $\{c_1\geq b_1,c_2\geq b_2\}$.
\end{itemize}
The result is illustrated in Figure \ref{fig:examples}(a).
\end{example}

\begin{example}\label{eg:lin}
Consider the distributions of the buyer's valuation to have a linear density. Specifically, consider
$$
  g_1(z_1)=\frac{1}{4}(1+z_1),\,\forall z_1\in[0,2],\quad g_2(z_2)=\frac{2}{5}(2+z_2),\,\forall z_2\in[0,1].
$$
We thus have $f_1(z_1)=\frac{1}{4}(1+z_1-c_1),\,\forall z_1\in[c_1,c_1+2]$, and $f_2(z_2)=\frac{2}{5}(2+z_2-c_2),\,\forall z_2\in[c_2,c_2+1]$. The thresholds $c_2^{th_k}(c_1)$ can be derived as
\begin{align*}
  c_2^{th_1}(c_1)&=\frac{45}{c_1^2-2c_1+10-(c_1-1)\sqrt{c_1^2-2c_1+28}}, \\
  c_2^{th_2}(c_1)&=\begin{cases}\frac{5}{4}\left(2+\frac{\tilde{\phi}_1^{-1}(0)}{\tilde{\phi}_1^{-1}(0)+(1-c_1)}\right)&\mbox{if }c_1\leq 1,\\\frac{5}{4}(2+c_1)&\mbox{if }c_1\in(1,4),\end{cases}
\end{align*}
where $\tilde{\phi}_1^{-1}(0)=\frac{2(c_1-1)+\sqrt{c_1^2-2c_1+28}}{3}$. The value of $c_2^{th_3}(c_1)$ can be computed by solving the following equation for $c_2$.
\begin{multline*}
  \frac{c_2}{30}\left(2\left(c_1+\sqrt{1+\frac{10}{c_2}}-1\right)^3-2c_1^3+3(1-c_1)\left(\left(c_1+\sqrt{1+\frac{10}{c_2}}-1\right)^2-c_1^2\right)\right) \\
  =\frac{1}{108}\left(\left(2(c_1-1)+\sqrt{c_1^2-2c_1+28}\right)^2\left((1-c_1)+\sqrt{c_1^2-2c_1+28}\right)\right).
\end{multline*}

The threshold functions $c_2^{th_k}(c_1)$, $k=1,2,3$, are plotted in Figure \ref{fig:thr-eg-2}(a). Observe that the final threshold $c_2^*(c_1)$ equals $c_2^{th_2}(c_1)$ for low values of $c_1$, and then equals $c_2^{th_3}(c_1)$ thereafter.

The thresholds $c_1^{th_k}(c_2)$ can be derived as
\begin{align*}
  c_1^{th_1}(c_2)&=\frac{180}{2c_2^2-8c_2-1-2(c_2-2)\sqrt{c_2^2-4c_2+31}}, \\
  c_1^{th_2}(c_2)&=4\left(2+\frac{\tilde{\phi}_2^{-1}(0)}{\tilde{\phi}_2^{-1}(0)+(2-c_2)}\right),
\end{align*}
where $\tilde{\phi}_2^{-1}(0)=\frac{2(c_2-2)+\sqrt{c_2^2-4c_2+31}}{3}$. The value of $c_1^{th_3}(c_2)$ can be computed by solving the following equation for $c_1$.
\begin{multline*}
  \frac{c_1}{60}\left(2\left(c_2+\sqrt{4+\frac{20}{c_1}}-2\right)^3-2c_2^3+3(2-c_2)\left(\left(c_2+\sqrt{4+\frac{20}{c_1}}-2\right)^2-c_2^2\right)\right) \\
  =\frac{2}{135}\left(\left(2(c_2-2)+\sqrt{c_2^2-4c_2+31}\right)^2\left((2-c_2)+\sqrt{c_2^2-4c_2+31}\right)\right).
\end{multline*}

The threshold functions $c_1^{th_k}(c_2)$, $k=1,2,3$, are plotted in Figure \ref{fig:thr-eg-2}(b). Observe that the final threshold $c_1^*(c_2)$ equals $c_1^{th_2}(c_2)$ for low values of $c_1$, and then equals $c_1^{th_3}(c_2)$ thereafter.

\begin{figure}
\centering
\begin{tabular}{cc}
\subfloat[]{\includegraphics[scale=0.175]{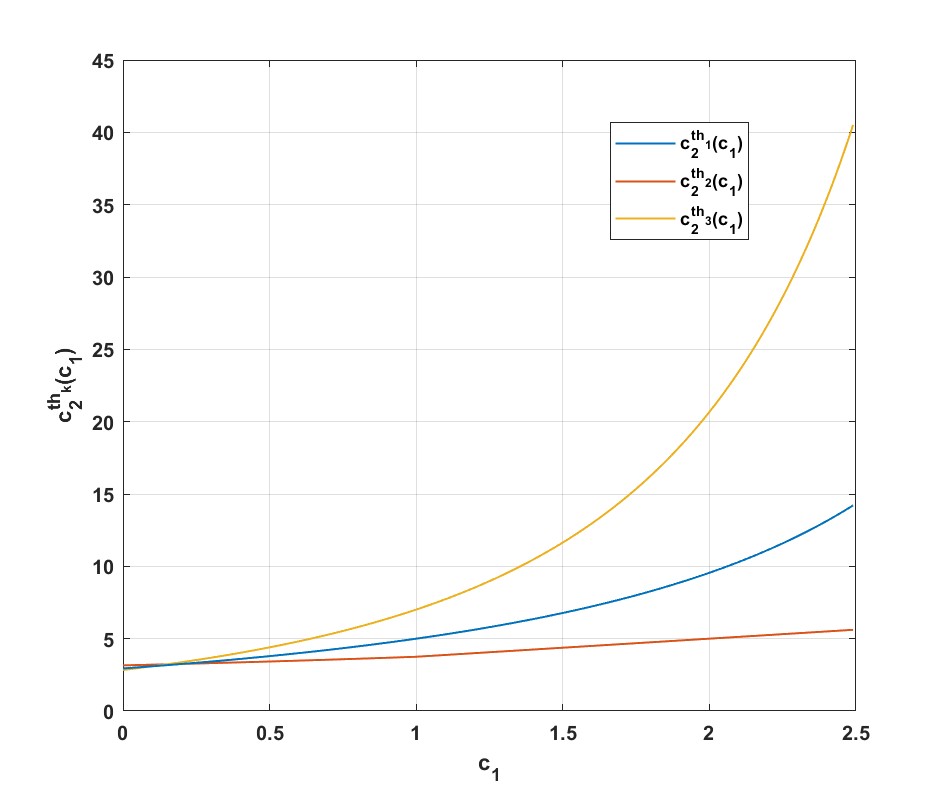}}&
\subfloat[]{\includegraphics[scale=0.175]{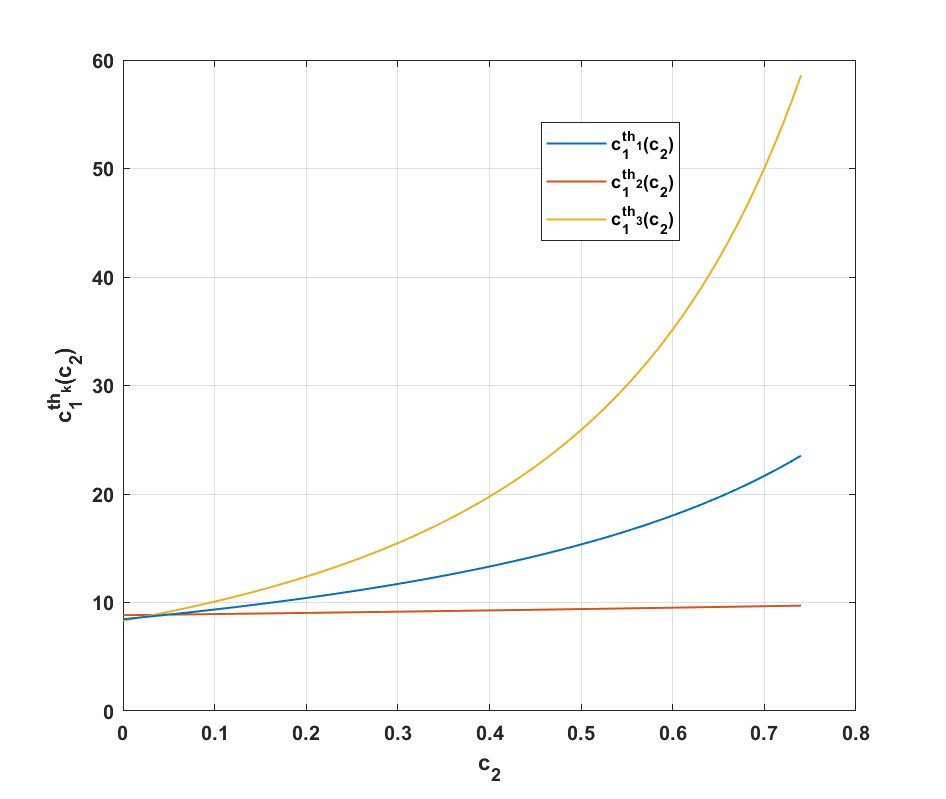}}
\end{tabular}
\caption{The threshold functions $c_2^{th_k}(c_1)$ and $c_1^{th_k}(c_2)$ when $g_1(z_1)=\frac{1}{4}(1+z_1)$ for $z_1\in[0,2]$ and $g_2(z_2)=\frac{2}{5}(2+z_2)$ for $z_2\in[0,1]$ respectively.}\label{fig:thr-eg-2}
\end{figure}

The detailed derivations of the thresholds are relegated to Appendix A. Now according to Theorem \ref{thm:individual-sale}, it is optimal for the seller to sell the items
\begin{itemize}
 \item at the price $\tilde{\phi}_1^{-1}(0)$ and $c_2$ respectively, when $\{c_1<4,c_2\geq c_2^*(c_1)\}$,
 \item at the price $c_1$ and $\tilde{\phi}_2^{-1}(0)$ respectively, when $\{c_1\geq c_1^*(c_2),c_2<\frac{5}{4}\}$,
 \item as a bundle when $\{c_1\geq 4,c_2\geq\frac{5}{4}\}$.
\end{itemize}
The result is illustrated in Figure \ref{fig:examples}(b).

\begin{figure}[h!]
\centering
\begin{tabular}{cc}
\subfloat[]{\includegraphics[scale=0.175]{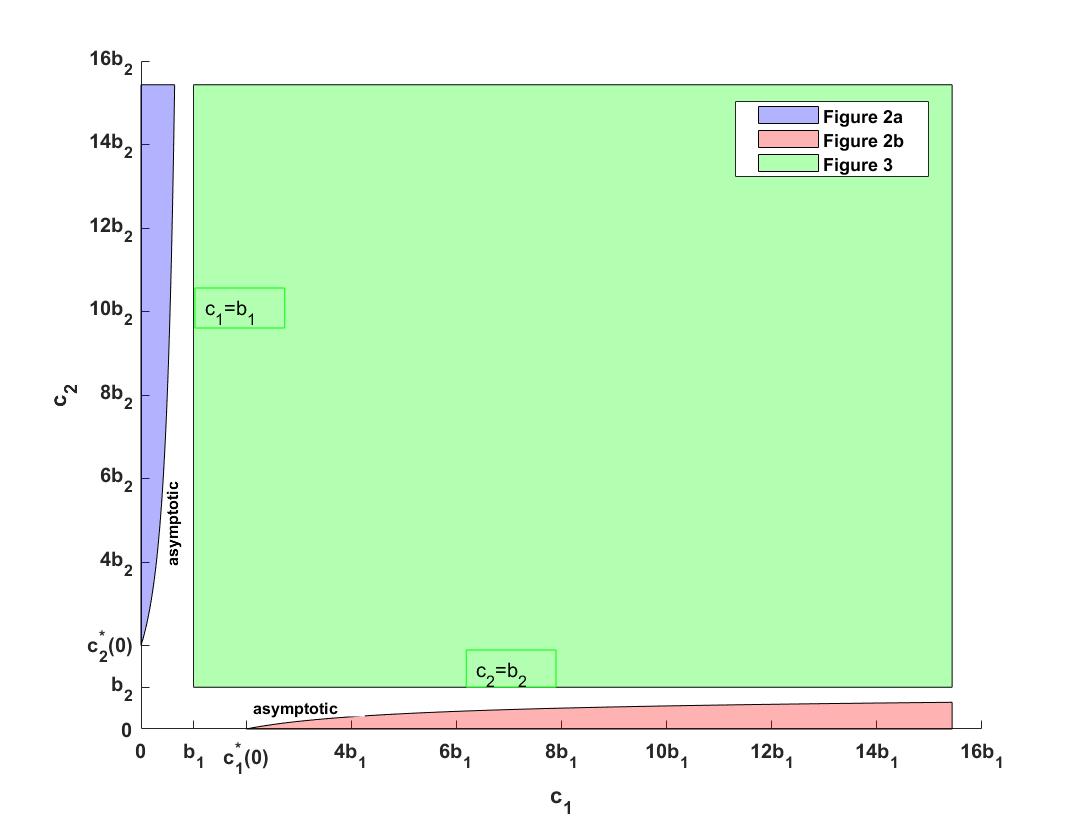}}&
\subfloat[]{\includegraphics[scale=0.175]{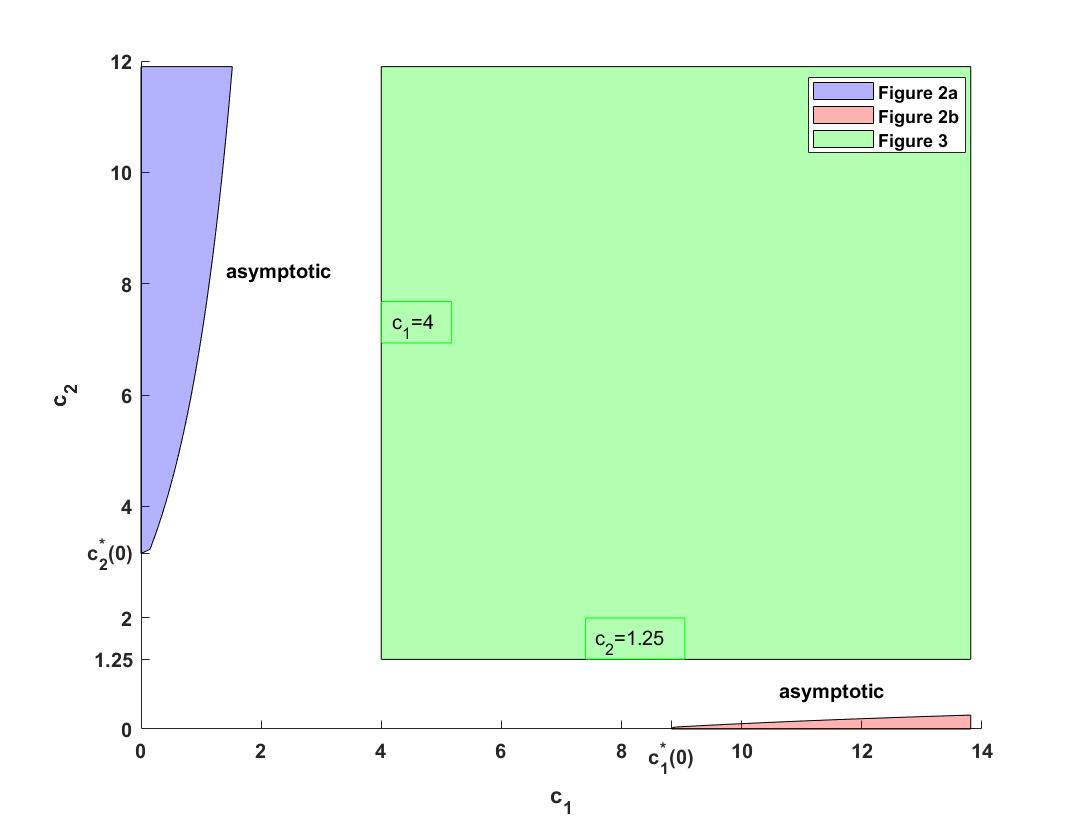}}
\end{tabular}
\caption{Optimal deterministic mechanisms for different values of $(c_1,c_2)$ when (a) $g_i(z_i)=\frac{1}{b_i}$ for $z_i\in[0,b_i]$, (b) $g_1(z_1)=\frac{1+z_1}{4}$ for $z_1\in[0,2]$ and $g_2(z_2)=\frac{2(2+z_2)}{5}$ for $z_2\in[0,1]$.}\label{fig:examples}
\end{figure}
\end{example}

\subsection{Individual Sale as the Optimal Mechanism}\label{sec:ind-sale}
An interesting question is whether it is optimal to sell the items individually for densities that are not nondecreasing. Putting Lemmas \ref{prop:W-region} and \ref{prop:A-region} together, we have the following result.
\begin{theorem}\label{thm:suff-individual}
Consider that the distribution of the buyer's valuation over the two items have densities $f_1$ and $f_2$, and that they are positive, continuously differentiable, and satisfy the following conditions.
\begin{enumerate}
 \item[(i)] $\mu(z)\leq0,\,\forall z\in D$,
 \item[(ii)] $\tilde{\phi}_1(\cdot)$ is strictly increasing, $\tilde{\phi}_1(c_1)<0$,
 \item[(iii)] $c_2f_2(c_2)f_1(z_1)\geq z_1f_1'(z_1)+2f_1(z_1),\,\forall z_1\in[c_1,\tilde{\phi}_1^{-1}(0)]$,
 \item[(iv)] There exists a $t\leq\tilde{\phi}_1^{-1}(0)-c_1$ such that $c_2f_2(c_2)F_1(c_1+t)=1$ and $c_2f_2(c_2)\left(\int_{c_1}^{c_1+t}z_1f_1(z_1)\,dz_1\right)\leq(\tilde{\phi}_1^{-1}(0))^2f_1(\tilde{\phi}_1^{-1}(0))$ hold.
\end{enumerate}
Then the optimal mechanism is to sell the items individually at prices $\tilde{\phi}_1^{-1}(0)$ and $c_2$ respectively. The structure of the optimal mechanism is as in Figure \ref{fig:e1}.
\end{theorem}

The following remarks interpret the conditions in the Theorem.
\begin{remark}
Define the power rate function\footnote{The quantity $\frac{zf'(z)}{f(z)}$ has been denoted by different terms in the literature. We use the term {\em power rate} as introduced in \citet{WT14}.} of the density $f_i$, $PR_{f_i}:[c_i,c_i+b_i]\rightarrow\mathbb{R}$, as $PR_{f_i}(z_i)=\frac{z_if_i'(z_i)}{f_i(z_i)}$. Observe that the condition (i) in Theorem \ref{thm:suff-individual} can be rewritten as $PR_{f_1}(z_1)+PR_{f_2}(z_2)+3\geq0$ for all $z\in D$, and the condition (ii), $\tilde{\phi}_1(\cdot)$ increasing or $\tilde{\phi}_1'(\cdot)>0$, can be rewritten as $PR_{f_1}(z_1)+2>0$ for all $z_1\in[c_1,c_1+b_1]$. The power rate of $f_i$ is nonnegative for nondecreasing $f_i$, and thus the two conditions are vacuously satisfied. However, these conditions may also be satisfied by $f_i$'s that are not nondecreasing, given that the conditions only ensure that $f_i$'s do not decrease too quickly.
\end{remark}

\begin{remark}
Consider that the densities $f_1$ and $f_2$, with their support sets $[c_1,c_1+b_1]$ and $[c_2,c_2+b_2]$ respectively, satisfy conditions (iii) and (iv). Then the same densities with support sets $[c_1,c_1+b_1]$ and $[\hat{c}_2,\hat{c}_2+b_2]$ also satisfy conditions (iii) and (iv) whenever $\hat{c}_2\geq c_2$. In other words, the densities $f_1$ and $f_2$ continue to satisfy conditions (iii) and (iv) even when the support set of $f_2$ is shifted to higher values. This is shown as a part of the proof of Theorem \ref{thm:individual-sale} in Appendix A. The proof for conditions (iii) and (iv) does not use the fact that the densities $f_1$ and $f_2$ are nondecreasing, and thus the remark is true for more general densities as well.
\end{remark}

The conditions in Theorem \ref{thm:suff-individual} thus require that
\begin{itemize}
 \item The densities $f_1$ and $f_2$ do not decrease too quickly so as to satisfy conditions (i) and (ii), and
 \item The value of $c_2$, the minimum valuation for item $2$, is sufficiently high so as to satisfy conditions (iii) and (iv).
\end{itemize}
The conditions (i)-(iv) are sufficient but not necessary to conclude that the optimal mechanism is to sell the items individually. I provide an example in Appendix C such that the densities $f_1$ and $f_2$ violate condition (iii) of Theorem \ref{thm:suff-individual} but the optimal mechanism still is to sell the items individually.

I now provide an example where $f_1$ and $f_2$ are {\em decreasing} but satisfy all the conditions in Theorem \ref{thm:suff-individual}.
\begin{example}\label{eg:inv-square}
Consider $f_1(z_1)=\frac{1}{2\sqrt{z_1}}$ when $z_1\in[0,1]$, and $f_2(z_2)=\frac{1}{2\sqrt{z_2}}$ when $z_2\in[9,16]$. So we have
$$
  \mu(z)=-f(z)\left(\frac{z_1f_1'(z_1)}{f_1(z_1)}+\frac{z_2f_2'(z_2)}{f_2(z_2)}+3\right)=-f(z)(-0.5-0.5+3)=-2f(z)\leq0.
$$
Condition (i) is thus satisfied. Now we have $F_1(z_1)=\sqrt{z_1}$, $F_2(z_2)=\sqrt{z_2}-3$ in their respective domains. So,
$$
  \tilde{\phi}_1(z_1)=z_1f_1(z_1)-(1-F_1(z_1))=\frac{3}{2}\sqrt{z_1}-1,
$$
is increasing, $\tilde{\phi}_1(c_1)=-1<0$, and $\tilde{\phi}_1^{-1}(0)=\frac{4}{9}$. Condition (ii) is thus satisfied. Now note that $c_2f_2(c_2)f_1(z_1)=\frac{3}{2}\cdot\frac{1}{2\sqrt{z_1}}$, and $z_1f_1'(z_1)+2f_1(z_1)=-\frac{z_1}{4z_1^{1.5}}+\frac{1}{\sqrt{z_1}}=\frac{3}{4\sqrt{z_1}}$. So we have
$$
  c_2f_2(c_2)f_1(z_1)=z_1f_1'(z_1)+2f_1(z_1),\,\forall z_1\in[c_1,c_1+b_1],
$$
and thus condition (iii) is satisfied.

Now $c_2f_2(c_2)F_1(c_1+t)=1$ when $F_1(t)=\frac{2}{3}$ or when $t=\frac{4}{9}=\tilde{\phi}_1^{-1}(0)-c_1$. Also, $(\tilde{\phi}_1^{-1}(0))^2f_1(\tilde{\phi}_1^{-1}(0))=(\frac{4}{9})^2\cdot\frac{3}{4}=\frac{4}{27}$, and $c_2f_2(c_2)\left(\int_{c_1}^{c_1+t}z_1f_1(z_1)\,dz_1\right)=\frac{3}{2}\int_0^{\frac{4}{9}}\frac{\sqrt{z_1}}{2}\,dz_1=\frac{3}{2}\cdot\frac{8}{27}\cdot\frac{1}{3}=\frac{4}{27}$. So we have
$$
  c_2f_2(c_2)\left(\int_{c_1}^{c_1+t}z_1f_1(z_1)\,dz_1\right)=(\tilde{\phi}_1^{-1}(0))^2f_1(\tilde{\phi}_1^{-1}(0)),
$$
and thus condition (iv) is also satisfied. The optimal mechanism thus is to sell the items at price $\frac{4}{9}$ and $9$ respectively.
\end{example}

We thus have an example where both $f_1$ and $f_2$ are decreasing but the optimal mechanism is to sell the items individually at prices $\tilde{\phi}_1^{-1}(0)$ and $c_2$. From Remark 3, conditions (iii) and (iv) hold true when for a fixed $\hat{c}_2\geq9$, the density $f_2$ is modified as $f_2(z_2)=\frac{1}{2\sqrt{z_2+9-\hat{c}_2}}$ for all $z_2\in[\hat{c}_2,\hat{c}_2+7]$. But condition (i) holds true only when $\hat{c}_2\leq45$. This is because
$$
  \mu(z)=-f(z)\left(-0.5-\frac{z_2}{2(z_2+9-\hat{c}_2)}+3\right)=-f(z)\left(2.5-\frac{z_2}{2(z_2+9-\hat{c}_2)}\right)\leq0
$$
is satisfied for all $z_2\in[\hat{c}_2,\hat{c}_2+7]$ only when $\hat{c}_2\leq45$. We thus know from Theorem \ref{thm:suff-individual} that the optimal mechanism is to sell the items individually when $\hat{c}_2\in[9,45]$. But it is not clear whether the optimal mechanism is deterministic when $\hat{c}_2>45$, given that the condition $\mu(z)\leq0$ is not satisfied for all $z\in D$.

Consider that the density $f_2$ is decreasing in an interval $[c_2+a,c_2+b]\subseteq[c_2,c_2+b_2]$. Then $f_2'(z_2)<0$ for all $z_2\in(c_2+a,c_2+b)$. We can thus choose a sufficiently large $\hat{c}_2$ such that when the support set of $f_2$ is $[\hat{c}_2,\hat{c}_2+b_2]$, we have
$$
  \mu(z)=\frac{z_1f_1'(z_1)}{f_1(z_1)}+\frac{z_2f_2'(z_2)}{f_2(z_2)}+3>0
$$
for some $z_2\in[\hat{c}_2+a,\hat{c}_2+b]$, since $f_2'(z_2)<0$ for all $z_2\in(\hat{c}_2+a,\hat{c}_2+b)$. So the condition $\mu(z)\leq0$ for all $z\in D$ fails to hold for a sufficiently large upward shift of the support set of $f_2$. The above argument thus indicates that the optimality of deterministic mechanisms for more general distributions, even if true, cannot be proved using the techniques employed in this section.

\subsection{Discussion}\label{sec:discussion}
I now discuss the interpretations of the results derived in this section.

{\em Tightness of the threshold $c_2^*(c_1)$:} The threshold $c_2^*(c_1)$ mentioned in Theorem \ref{thm:individual-sale} is not tight in general. In other words, there exist distributions where the threshold beyond which the optimal mechanism is to sell the two heterogeneous items individually are lower than the threshold $c_2^*(c_1)$. I establish this in Appendix C using a counter-example. Specifically, I consider the same densities in Example \ref{eg:lin}, and show that $c_2^*(0)\approx3.154$ is not tight; the optimal mechanism is to sell the items individually even when $c_2\geq3$. Deriving the exact threshold values, even in the case of positive, nondecreasing and continuously differentiable densities, remains an open problem.

{\em Extension to other distributions:} Regarding the optimality of individual sale mechanism, a relevant question is whether the optimality holds for general distributions when $c_1$ or $c_2$ is increased indefinitely. I now provide a counter-example to show that the optimality result need not hold for all distributions.

Let $g_i\sim\mbox{Beta}(3,3)$, i.e., $g_i(z_i)=30z_i^2(1-z_i)^2\mathbf{1}\{z_i\in[0,1]\}$, $i=1,2$. The density decreases when $z_i\geq0.5$, and thus does not satisfy the assumptions in this section. Neither is the density positive since $g_i(0)=g_i(1)=0$. At $c_1=0$, we have $\tilde{\phi}_1^{-1}(0)\approx0.3981$, and thus the individual sale mechanism would allocate both items to the buyer when $z_1\geq0.3981$, and allocate only item $2$ otherwise. Let $A=[0,\tilde{\phi}_1^{-1}(0))\times[c_2,c_2+b_2]$. I now show that $\bar{\mu}|_A\npreceq_{cvx(\overrightarrow{1,-1})}0$ for any $c_2\geq0$. I show this by considering a convex $u$ that is nondecreasing in $z_1$ and nonincreasing in $z_2$, and proving that $\int_Au\,d\bar{\mu}>0$ for any $c_2$. Let $u(z_1,z_2)=\max(0,c_2+1/c_2-z_2)$. Now
$$
  \int_Au\,d\bar{\mu}=c_2^5-(3.142)c_2^4+(4.713)c_2^3-(6.74)c_2^2+(6.77)c_2-(2.571)>0
$$
for all $c_2\geq1.5$. Also, if $u(z_1,z_2)=\max(0,c_2+2/3-z_2)$, we have
$$
  \int_Au\,d\bar{\mu}=(0.3732)-(0.248)c_2>0
$$
for all $c_2\leq1.5$. So at $c_1=0$, the individual sale mechanism is not optimal for any $c_2$.

We thus have a counter-example to show that the optimality result need not hold for all distributions. However, it may hold for more general distributions than those having a positive, nondecreasing and continuously differentiable density, as indicated by Example \ref{eg:inv-square}. Mathematical analysis to show necessary conditions on such distributions are notoriously hard given that it involves verification of second-order stochastic dominance conditions in two-dimensional space.

{\em Extension to the unit-demand setting:} Another interesting question is whether the property of optimal mechanisms being deterministic for sufficiently high minimum valuations holds true in the two-item unit-demand setting as well. I now provide an answer based on the exact optimal mechanisms known in the literature for different distributions.

Consider the setting when the valuations of the buyer, $z_1$ and $z_2$, are independently and identically distributed according to $f$, where $f=$Unif$[c,c+1]$. The result in \citet[Example.~1]{Pavlov11} shows that the optimal mechanism is {\em not deterministic} when $c>1$. Thus the optimal mechanism need not be deterministic for sufficiently high minimum valuations even in the simplest setting of identical distributions.

Now consider a setting when $z_1$ and $z_2$ are independent but not identically distributed. Specifically, let $z_1\sim\mbox{Unif}[c,c+b_1]$ and $z_2\sim\mbox{Unif}[c,c+b_2]$, with $b_1\ne b_2$. The result in \citet[Thm.~12]{Thiru-Unit-19} shows that the optimal mechanism in this setting is {\em deterministic} for sufficiently high minimum valuations.

The contrasting set of results suggest that the property of optimal mechanisms being deterministic for high minimum valuations cannot be extended in a straight-forward manner to the unit-demand setting. A more careful analysis is required to characterize the distributions for which the property holds in the unit-demand setting.

\section{Extension to Three-Item Setting}\label{sec:ext}
I now consider an extension where three heterogeneous items are sold to a single buyer. So I consider three distributions whose densities are given by $g_i:[0,b_i]\rightarrow\mathbb{R}_+$, $i=1,2,3$, each of which is strictly positive, nondecreasing and continuously differentiable in its domain. The distribution of buyer's valuation for item $i$ is given by
$$
  f_i(z_i)=\begin{cases}g_i(z_i-c_i)&\mbox{if }z_i\in[c_i,c_i+b_i],\\0&\mbox{else},\end{cases}
$$
for some $c_i\geq0$. I now analyze how the optimal mechanism varies with $(c_1,c_2,c_3)$. I begin by observing that the properties of $\mu$ and $\tilde{\phi}_i$ proved in Lemma \ref{lem:mu-phi-i} are true in the three-item setting as well. I skip the proof given that it is exactly the same as the proof of Lemma \ref{lem:mu-phi-i}.

I am now ready to state the main results in the three-item setting. The proofs are relegated to Appendix B.
\begin{theorem}\label{thm:three-high}
Consider that the minimum valuations for the items, $c_1,c_2,c_3$, are such that $c_i\geq\frac{1}{g_i(0)}$. Then the optimal mechanism is to sell all the three items as a bundle. The structure of the optimal mechanism is as depicted in Figure \ref{fig:three-item}(a).
\end{theorem}
\begin{theorem}\label{thm:one-high}
Consider the setting where the valuations of the buyer, $(z_1,z_2,z_3)$, are distributed as $z_1\sim\mbox{Unif}[0,1]$, $z_2\sim\mbox{Unif}[0,1]$, and $z_3\sim\mbox{Unif}[c_3,c_3+1]$, with $c_3\geq3$. Then the optimal mechanism is as follows:
\begin{itemize}
 \item Item $3$ is sold individually for a price of $c_3$.
 \item Items $1$ and $2$ are bundled together and sold as follows.
 \begin{itemize}
  \item Item $1$ is sold for a price of $\frac{2}{3}$.
  \item Item $2$ is sold for a price of $\frac{2}{3}$.
  \item The bundle having items $1$ and $2$ is sold for a price of $\frac{4-\sqrt{2}}{3}$.
 \end{itemize}
\end{itemize}
 The structure of the optimal mechanism is as depicted in Figure \ref{fig:three-item}(b).\end{theorem}

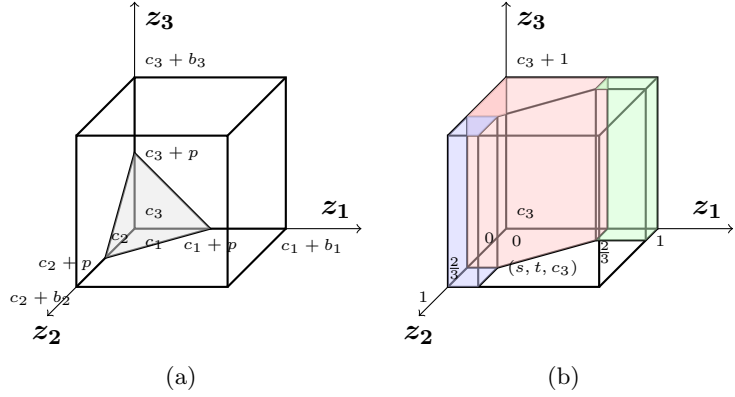
\begin{figure}
\centering
\begin{tabular}{cc}
\subfloat[]{\begin{tikzpicture}[scale=2]
\draw[->] (0,0,0) -- (1.5,0,0) node[anchor=south east]{$\bm{z_1}$};
\draw[->] (0,0,0) -- (0,0,1.5) node[anchor=north]{$\bm{z_2}$};
\draw[->] (0,0,0) -- (0,1.5,0) node[anchor=north west]{$\bm{z_3}$};

\draw[thick] (0,0,0)--(1,0,0)--(1,1,0)--(0,1,0)--cycle;
\draw[thick] (0,0,1)--(1,0,1)--(1,1,1)--(0,1,1)--cycle;
\draw[thick] (0,0,0)--(0,0,1);
\draw[thick] (1,0,0)--(1,0,1);
\draw[thick] (1,1,0)--(1,1,1);
\draw[thick] (0,1,0)--(0,1,1);

\draw[thick] (0,0,0.5)--(0,0.5,0)--(0.5,0,0)--cycle;
\fill[gray!20, opacity=0.5] (0,0,0.5)--(0,0.5,0)--(0.5,0,0)--cycle;

\node[anchor=north west] at (0,0,0) {\tiny$c_1$};
\node[anchor=north east] at (0,0,-0.1) {\tiny$c_2$};
\node[anchor=south west] at (0,0,0) {\tiny$c_3$};
\node[anchor=north] at (0.5,0,0) {\tiny$c_1+p$};
\node[anchor=north east] at (-0.1,0,0.3) {\tiny$c_2+p$};
\node[anchor=west] at (0,0.5,0) {\tiny$c_3+p$};
\node[anchor=north west] at (0.9,0,0) {\tiny$c_1+b_1$};
\node[anchor=north east] at (0,0,0.9) {\tiny$c_2+b_2$};
\node[anchor=south west] at (0,1,0) {\tiny$c_3+b_3$};
\end{tikzpicture}}&

\subfloat[]{\begin{tikzpicture}[scale=2]
\draw[->] (0,0,0) -- (1.5,0,0) node[anchor=south east]{$\bm{z_1}$};
\draw[->] (0,0,0) -- (0,0,1.5) node[anchor=north]{$\bm{z_2}$};
\draw[->] (0,0,0) -- (0,1.5,0) node[anchor=north west]{$\bm{z_3}$};

\draw[thick] (0,0,0)--(1,0,0)--(1,1,0)--(0,1,0)--cycle;
\draw[thick] (0,0,1)--(1,0,1)--(1,1,1)--(0,1,1)--cycle;
\draw[thick] (0,0,0)--(0,0,1);
\draw[thick] (1,0,0)--(1,0,1);
\draw[thick] (1,1,0)--(1,1,1);
\draw[thick] (0,1,0)--(0,1,1);

\draw[thick] (0.67,0,0)--(0.67,0,0.2)--(0.2,0,0.67)--(0,0,0.67);
\draw[thick] (0.67,1,0)--(0.67,1,0.2)--(0.2,1,0.67)--(0,1,0.67);
\draw[thick] (0.67,0,0)--(0.67,1,0);
\draw[thick] (1,0,0.2)--(0.67,0,0.2);
\draw[thick] (0.2,0,1)--(0.2,0,0.67);
\draw[thick] (0.67,0,0.2)--(0.67,1,0.2);
\draw[thick] (1,0,0.2)--(1,1,0.2);
\draw[thick] (0.2,0,0.67)--(0.2,1,0.67);
\draw[thick] (0,0,0.67)--(0,1,0.67);
\draw[thick] (0.2,0,1)--(0.2,1,1);
\draw[thick] (1,1,0.2)--(0.67,1,0.2);
\draw[thick] (0.2,1,1)--(0.2,1,0.67);
\fill[red!20, opacity=0.5] (0.67,0,0)--(0.67,0,0.2)--(0.2,0,0.67)--(0,0,0.67)--(0,0,0)--(0,1,0)--(0,1,0.67)--(0.2,1,0.67)--(0.67,1,0.2)--(0.67,1,0)--cycle;
\fill[red!20, opacity=0.5] (0,1,0)--(0,1,0.67)--(0.2,1,0.67)--(0.67,1,0.2)--(0.67,1,0)--cycle;
\fill[red!20, opacity=0.5] (0,0,0)--(0,0,0.67)--(0,1,0.67)--(0,1,0)--cycle;
\fill[green!20, opacity=0.5] (0.67,0,0)--(0.67,0,0.2)--(1,0,0.2)--(1,0,0)--(1,1,0)--(1,1,0.2)--(0.67,1,0.2)--(0.67,1,0)--cycle;
\fill[green!20, opacity=0.5] (1,1,0)--(1,1,0.2)--(0.67,1,0.2)--(0.67,1,0)--cycle;
\fill[blue!20, opacity=0.5] (0,0,0.67)--(0.2,0,0.67)--(0.2,0,1)--(0,0,1)--(0,1,1)--(0.2,1,1)--(0.2,1,0.67)--(0,1,0.67)--cycle;
\fill[blue!20, opacity=0.5] (0,1,1)--(0.2,1,1)--(0.2,1,0.67)--(0,1,0.67)--cycle;

\node[anchor=north west] at (-0.05,0,-0.05) {\tiny$0$};
\node[anchor=north east] at (-0.05,0,-0.1) {\tiny$0$};
\node[anchor=south west] at (0,0,0) {\tiny$c_3$};
\node[anchor=north] at (0.67,0,0.05) {\tiny$\frac{2}{3}$};
\node[anchor=north east] at (-0.1,0,0.32) {\tiny$\frac{2}{3}$};
\node[anchor=north west] at (0.9,0,-0.05) {\tiny$1$};
\node[anchor=north east] at (-0.1,0,0.9) {\tiny$1$};
\node[anchor=west] at (0.2,0,0.67) {\tiny$(s,t,c_3)$};
\node[anchor=south west] at (0,1,0) {\tiny$c_3+1$};
\end{tikzpicture}}
\end{tabular}
\caption{(a) The structure of the optimal mechanism when $c_i\geq\frac{1}{g_i(0)}$, $i=1,2,3$. The optimal allocation is $q(z)=(0,0,0)$ in the gray region, and $q(z)=(1,1,1)$ in the other region. (b) The structure of the optimal mechanism when $z_1\sim\mbox{Unif}[0,1]$, $z_2\sim\mbox{Unif}[0,1]$, and $z_3\sim\mbox{Unif}[c_3,c_3+1]$, with $c_3\geq3$. The point $(s,t,c_3)=(\frac{2-\sqrt{2}}{3},\frac{2}{3},c_3)$. The optimal allocation is $q(z)=(0,0,1)$ in the red region, $q(z)=(1,0,1)$ in the green region, $q(z)=(0,1,1)$ in the blue region, and $q(z)=(1,1,1)$ in the remaining region.}\label{fig:three-item}
\end{figure}

In other words, it is optimal for the seller to sell item $3$ with a ``high'' minimum valuation at the minimum valuation $c_3$, and items $1$ and $2$ with ``low'' minimum valuations according to the optimal mechanism in the two-item setting. Using this example, I conjecture that when $c_1$ and $c_2$ are low but $c_3$ is high, the property that the optimal mechanism sells item $3$ at the minimum valuation $c_3$ and the other two items according to the optimal mechanism in the two-dimensional setting holds in three-item setting as well.

Theorems \ref{thm:three-high} and \ref{thm:one-high} indicate that some of the results derived in Section \ref{sec:main-results} in the two-item setting extend to the three-item setting as well, whereas the other results do not. Specifically, the property of
\begin{itemize}
 \item bundle sale being optimal when the minimum valuations of all the items are high {\em extends} to the three-item setting,
 \item the problem being reduced to finding the optimal mechanism in a lower dimension when the minimum valuation of one of the items is high and the other two items is low {\em can possibly extend} to the three-item setting, and
 \item optimal mechanisms being deterministic when the minimum valuation of one of the items is high and the other two items is low {\em may not extend} to the three-item setting, given that the optimal mechanism in the two-item setting need not be deterministic.
\end{itemize}

\section{Conclusion}\label{sec:conclusion}
I considered the problem of selling two heterogeneous items to a single buyer whose valuation for the items $(z_1,z_2)\sim f_1f_2$, and the support set of $f_i$ is an interval $[c_i,c_i+b_i]$ in the nonnegative axis of $z_i$. I restricted attention to density functions $f_1$ and $f_2$ that are positive, nondecreasing and continuously differentiable. I analyzed how the optimal mechanism varies when the distribution is fixed but the support set parameters $(c_1,c_2)$ vary.

I showed that the optimal mechanisms are deterministic if the minimum valuations of at least one of the items (i.e., either $c_1$ or $c_2$) is sufficiently high. In particular, I showed that the seller finds it optimal to sell the items individually when $c_1$ is lower than the threshold $c_1^{th}$ and $c_2$ is at least $c_2^*(c_1)$, or when $c_2$ is lower than the threshold $c_2^{th}$ and $c_1$ is at least $c_1^*(c_2)$; and he finds it optimal to sell the items as a bundle when both $c_1$ and $c_2$ are at least $c_1^{th}$ and $c_2^{th}$ respectively. In addition, I provided a method to calculate the threshold function $c_2^*(c_1)$ at each value of $c_1$, and the function $c_1^*(c_2)$ at each value of $c_2$. Furthermore, I also provided verifiable sufficient conditions on the densities $f_1$ and $f_2$ for which the individual sale mechanism is optimal.

I showed that when $c_1$ is low and $c_2$ is high, the seller finds it optimal to sell item $2$ at the minimum valuation $c_2$, thus effectively reducing the problem to finding the optimal mechanism in the one-item setting only for item $1$. I then considered an example in the three-item setting with the items distributed uniformly in the cube $[0,1]\times[0,1]\times[c_3,c_3+1]$, and showed that the seller finds it optimal to sell item $3$ at the minimum valuation $c_3$, and items $1$ and $2$ according to the optimal mechanism in the two-dimensional setting. I conjectured that the reduction to $(n-1)$-dimensional setting can possibly be extended to three-item setting for more general distributions.

One of the key contributions of this paper is in finding the exact optimal mechanism for a wide class of distributions whose support sets do not contain zero. The exact optimal mechanism is known in the literature only for a few such distributions, the key reason being the difficulty in proving the second-order stochastic dominance of measures defined on a multi-dimensional space. To the best of my knowledge, there are no straightforward methods available in the literature to verify second order stochastic dominance of a given measure over another measure in multi-dimensional settings. The key technical contribution of this paper is in showing the second-order dominance for a class of distributions having a positive, nondecreasing and continuously differentiable densities in both two-dimensional and three-dimensional settings.

The result that the optimal mechanism is deterministic when minimum valuations are sufficiently high can possibly be true for more general distributions. But the arguments in this paper suggest that proving such a result might be more technically involved primarily because of the difficulty in establishing the second order stochastic dominance of measures. Characterizing the optimal mechanisms for more general distributions in higher dimensional settings is a direction of work for the future.

\section*{Acknowledgements}
This work was supported by the Science and Engineering Research Board [Sanction Order No. CRG/2022/009169] under the Core Research Grant (CRG) Program; and the Indian Institute of Technology Kanpur [Project No. IITK/ECOS/2020123] under the initiation grant scheme. The author thanks Prof. Debasis Mishra for a very informative discussion.

\bibliographystyle{elsarticle-harv.bst} \bibliography{deterministic_mechanisms}

\section*{APPENDIX}
\section*{A. Proofs in Section \ref{sec:main-results}}
{\bf Proof of Lemma \ref{lem:mu-phi-i}:} (a) We know that $\mu(z)=-z_1f_1'(z_1)f_2(z_2)-z_2f_1(z_1)f_2'(z_2)-3f_1(z_1)f_2(z_2)$. So $\mu(z)<0$ follows since $z_i\geq0$, $f_i'(z_i)\geq0$, and $f_i(z_i)>0$ for $i=1,2$.

(b) Differentiating $\tilde{\phi}_i(\cdot)$, we have
$$
  \tilde{\phi}_i'(z_i)=z_if_i'(z_i)+2f_i(z_i)>0,
$$
since $z_i\geq0$, $f_i'(z_i)\geq0$ and $f_i(z_i)>0$. Thus  $\tilde{\phi}_i$ is strictly increasing.

(c) Note that $\tilde{\phi}_i(c_i)=c_if_i(c_i)-1=c_ig_i(0)-1<0$ if $c_i<\frac{1}{g_i(0)}$. Furthermore, $\tilde{\phi}_i(c_i+b_i)=(c_i+b_i)f_i(c_i+b_i)>0$ since $f_i>0$. The existence and uniqueness of $\tilde{\phi}_i^{-1}(0)$ then follows since $\tilde{\phi}_i(\cdot)$ is continuous and strictly increasing. Furthermore, $\tilde{\phi}_i^{-1}(0)\in(c_i,c_i+b_i)$ since $\tilde{\phi}_i(c_i)<0$ and $\tilde{\phi}_i(c_i+b_i)>0$. \qed

{\bf Proof of Lemma \ref{prop:c2-th-3}:} I first show that $\frac{\int_{c_1}^{c_1+t(c_2)}z_1f_1(z_1)\,dz_1}{F_1(c_1+t(c_2))}$ is decreasing in $c_2$. Differentiating with respect to $t$, we have
\begin{multline*}
  \frac{F_1(c_1+t)(c_1+t)f_1(c_1+t)-f_1(c_1+t)\int_{c_1}^{c_1+t}z_1f_1(z_1)\,dz_1}{(F_1(c_1+t))^2} \\
  =\frac{f_1(c_1+t)\int_{c_1}^{c_1+t}(c_1+t-z_1)f_1(z_1)\,dz_1}{(F_1(c_1+t))^2}>0,
\end{multline*}
where the inequality follows because $f_i(z_i)>0$. So the term $\frac{\int_{c_1}^{c_1+t}z_1f_1(z_1)\,dz_1}{F_1(c_1+t)}$ increases as $t$ increases. Observe that $t$ decreases as $c_2$ increases, since $t(c_2)=\{t:c_2g_2(0)F_1(c_1+t)=1\}$. Thus the term $\frac{\int_{c_1}^{c_1+t(c_2)}z_1f_1(z_1)\,dz_1}{F_1(c_1+t(c_2))}$ decreases as $c_2$ increases.

Now note that $t\left(\frac{1}{g_2(0)}\right)=b_1$, and $\lim_{c_2\rightarrow\infty}t(c_2)=0$. This shows that the range of $t(\cdot)$ is $(0,b_1]$. Now we have
\begin{align*}
  \frac{\int_{c_1}^{c_1+t\left(\frac{1}{g_2(0)}\right)}z_1f_1(z_1)\,dz_1}{F_1\left(c_1+t\left(\frac{1}{g_2(0)}\right)\right)}&=\mathbb{E}_{z_1\sim f_1}[z_1], \\
  \lim_{c_2\rightarrow\infty}\frac{\int_{c_1}^{c_1+t(c_2)}z_1f_1(z_1)\,dz_1}{F_1(c_1+t(c_2))}&\stackrel{(H)}{=}\lim_{c_2\rightarrow\infty}\frac{(c_1+t(c_2))f_1(c_1+t(c_2))}{f_1(c_1+t(c_2))}=c_1,
\end{align*}
where $(H)$ refers to equality based on L'H\^{o}pital's rule. But $\frac{\int_{c_1}^{c_1+t(c_2)}z_1f_1(z_1)\,dz_1}{F_1(c_1+t(c_2))}$ decreases as $c_2$ increases, and thus
$c_1<\frac{\int_{c_1}^{c_1+t(c_2)}z_1f_1(z_1)\,dz_1}{F_1(c_1+t(c_2))}\leq\mathbb{E}_{z_1\sim f_1}[z_1]$. Furthermore, there exists a unique $c_2\in\left[\frac{1}{g_2(0)},\infty\right)$ such that
$$
  \frac{\int_{c_1}^{c_1+t(c_2)}z_1f_1(z_1)\,dz_1}{F_1(c_1+t(c_2))}=c\mbox{ for any }c\in(c_1,\mathbb{E}_{z_1\sim f_1}[z_1]].
$$

The proof is complete if I show that $\max_{z_1\in[c_1,c_1+b_1]}(z_1(1-F_1(z_1)))\in(c_1,\mathbb{E}_{z_1\sim f_1}[z_1]]$. Observe from Theorem \ref{thm:myerson} that $z_1(1-F_1(z_1))$ attains its maximum at $z_1^*=\tilde{\phi}_1^{-1}(0)$ since $\tilde{\phi}_1(\cdot)$ is strictly increasing, and from Lemma \ref{lem:mu-phi-i}(c) that $\tilde{\phi}_1^{-1}(0)>c_1$. So $\tilde{\phi}_1(c_1)<0$. Now we have
\begin{multline*}
  \frac{d}{dz_1}z_1(1-F_1(z_1))|_{z_1=c_1}=-\tilde{\phi}_1(c_1)>0\Rightarrow c_1^+(1-F_1(c_1^+))>c_1 \\
  \Rightarrow\max_{z_1\in[c_1,c_1+b_1]}z_1(1-F_1(z_1))>c_1.
\end{multline*}

I now show that $\max_{z_1\in[c_1,c_1+b_1]}(z_1(1-F_1(z_1)))\leq\mathbb{E}_{z_1\sim f_1}[z_1]$. Observe that
\begin{multline*}
  z_1(1-F_1(z_1))=(1-F_1(z_1))\int_0^{z_1}\,dz\stackrel{(a)}{\leq}\int_0^{z_1}(1-F_1(z))\,dz \\
  \leq\int_0^{\infty}(1-F_1(z))\,dz=\mathbb{E}_{z_1\sim f_1}[z_1],
\end{multline*}
where (a) follows since $(1-F_1(z_1))$ is nonincreasing. Hence the result. \qed

{\bf Proof of Lemma \ref{prop:W-region}:} Observe first that $\tilde{\phi}_1^{-1}(0)$ exists since (i) $\tilde{\phi}_1(c_1)<0$ and $\tilde{\phi}_1(c_1+b_1)=(c_1+b_1)f_1(c_1+b_1)>0$, and (ii) $\tilde{\phi}_1$ is continuous since $f_1$ is continuous. So the set $W=[\tilde{\phi}_1^{-1}(0),c_1+b_1]\times[c_2,c_2+b_2]$ is well-defined.

Now, note that one can prove $\bar{\mu}|_W\succeq_10$ by proving that (i) $\bar{\mu}|_W(X)\geq 0$ for any increasing set\footnote{A set $X\subseteq W$ is an increasing set if $(x_1,x_2)\in X$ implies $(y_1,y_2)\in X$ for all $\{(y_1,y_2)\in W:y_1\geq x_1,y_2\geq x_2\}$.} $X\subseteq W$, and (ii) $\bar{\mu}(W)=0$ \cite[Chap.~6]{SS07}.Consider the sets of the form $[c_1+t_1,c_1+b_1]\times[c_2+t_2,c_2+b_2]$ for some $t_1\in[\tilde{\phi}_1^{-1}(0)-c_1,b_1]$ and $t_2\in(0,b_2]$. Observe that
$$
  \int(-z_1f_1'(z_1)f_2(z_2)-2f_1(z_1)f_2(z_2))\,dz_1=-(z_1f_1(z_1)+F_1(z_1))f_2(z_2),
$$
and
$$
  \int(-z_2f_1(z_1)f_2'(z_2)-f_1(z_1)f_2(z_2))\,dz_2=-z_2f_2(z_2)f_1(z_1).
$$
So we have
\begin{align}
  &\bar{\mu}|_W([c_1+t_1,c_1+b_1]\times[c_2+t_2,c_2+b_2]) \nonumber\\
  &=\int_{c_2+t_2}^{c_2+b_2}\int_{c_1+t_1}^{c_1+b_1}(-z_1f_1'(z_1)f_2(z_2)-z_2f_2'(z_2)f_1(z_1)-3f_1(z_1)f_2(z_2))\,dz_1\,dz_2 \nonumber\\
  &\hspace*{.2in}+(c_1+b_1)f_1(c_1+b_1)\int_{c_2+t_2}^{c_2+b_2}f_2(z_2)\,dz_2+(c_2+b_2)f_2(c_2+b_2)\int_{c_1+t_1}^{c_1+b_1}f_1(z_1)\,dz_1 \nonumber\\
  &=-(1-F_2(c_2+t_2))((c_1+b_1)f_1(c_1+b_1)+1-(c_1+t_1)f_1(c_1+t_1)-F_1(c_1+t_1)) \nonumber\\
  &\hspace*{.2in}-(1-F_1(c_1+t_1))((c_2+b_2)f_2(c_2+b_2)-(c_2+t_2)f_2(c_2+t_2)) \nonumber\\
  &\hspace*{.3in}+(c_1+b_1)f_1(c_1+b_1)(1-F_2(c_2+t_2))+(c_2+b_2)f_2(c_2+b_2)(1-F_1(c_1+t_1)) \nonumber\\
  &=\tilde{\phi}_1(c_1+t_1)(1-F_2(c_2+t_2))+(c_2+t_2)f_2(c_2+t_2)(1-F_1(c_1+t_1)) \label{eqn:mu-bar-W}\\
  &\geq 0,\nonumber
\end{align}
where the last step follows since (i) $c_1+t_1\geq\tilde{\phi}_1^{-1}(0)$ and (ii) $\tilde{\phi}_1$ increasing imply that $\tilde{\phi}_1(c_1+t_1)\geq0$.

I now consider sets of the form $[c_1+t_1,c_1+b_1]\times[c_2,c_2+b_2]$ for some $t_1\in[\tilde{\phi}_1^{-1}(0)-c_1,b_1]$.
\begin{align*}
  &\bar{\mu}|_W([c_1+t_1,c_1+b_1]\times[c_2,c_2+b_2]) \\
  &=\lim_{t_2\rightarrow 0}\left(\bar{\mu}|_W([c_1+t_1,c_1+b_1]\times[c_2+t_2,c_2+b_2])\right)-c_2f_2(c_2)(1-F_1(c_1+t_1) \\
  &\stackrel{(a)}{=}\tilde{\phi}_1(c_1+t_1)(1-F_2(c_2))+c_2f_2(c_2)(1-F_1(c_1+t_1))-c_2f_2(c_2)(1-F_1(c_1+t_1) \\
  &=\tilde{\phi}_1(c_1+t_1) \\
  &\geq 0,
\end{align*}
where (a) follows from \eqref{eqn:mu-bar-W}. Note that
$$
  \bar{\mu}(W)=\bar{\mu}|_W([\tilde{\phi}_1^{-1}(0),c_1+b_1]\times[c_2,c_2+b_2])=\tilde{\phi}_1(\tilde{\phi}_1^{-1}(0))=0.
$$

I have thus shown that $\bar{\mu}|_W(X)\geq 0$ for any increasing rectangle $X\subseteq W$. I now extend this result for increasing sets that are not rectangular. Let $c_1+t_1=\min\{z_1:(z_1,c_2+b_2)\in X\}$, and $c_2+t_2=\min\{z_2:(c_1+b_1,z_2)\in X\}$. Observe that $X\subseteq[c_1+t_1,c_1+b_1]\times[c_2+t_2,c_2+b_2]$, since $X$ is an increasing set. Let $Y=([c_1+t_1,c_1+b_1]\times[c_2+t_2,c_2+b_2])\backslash X$. Now we have
\begin{multline*}
  \bar{\mu}|_W(X)=\bar{\mu}|_W([c_1+t_1,c_1+b_1]\times[c_2+t_2,c_2+b_2])-\int\int_Y\mu(z)\,dz_1\,dz_2 \\
  \geq\bar{\mu}|_W([c_1+t_1,c_1+b_1]\times[c_2+t_2,c_2+b_2])\geq0,
\end{multline*}
where the first inequality occurs because $\mu(z)\leq0$ for all $z\in D$. Thus $\bar{\mu}|_W(X)\geq0$ for any increasing set $X\subseteq W$, and this implies $\bar{\mu}|_W\succeq_10$.

Observe that $\bar{\mu}|_W\succeq_10$ implies $\bar{\mu}|_W\preceq_{cvx(\overrightarrow{-1,-1})}0$, since $\int_Wu\,d\bar{\mu}\geq0$ for all increasing $u$ implies that $\int_Wu\,d\bar{\mu}\leq0$ for all convex, decreasing $u$. \qed

{\bf Proof of Lemma \ref{prop:A-region}:} I prove $\bar{\mu}|_A\preceq_{cvx\overrightarrow{(1,-1)}}0$ by showing that $\int_Au\,d\bar{\mu}\leq0$ for all $u$ that is (i) convex, (ii) nondecreasing in $z_1$, and (iii) nonincreasing in $z_2$. Define $B=[c_1,c_1+t]\times\{c_2\}$. I first show that (i) $\bar{\mu}(B)=0$, (ii) $\bar{\mu}(A\backslash B)=0$, and (iii) $\int_{A\backslash B}(z_1-c_1)\,d\bar{\mu}\leq 0$. I then show that this is equivalent to showing $\bar{\mu}|_A\preceq_{cvx(\overrightarrow{1,-1})}0$.

I begin by proving $\bar{\mu}(B)=0$.
$$
  \bar{\mu}(B)=1-\int_{c_1}^{c_1+t}c_2f_2(c_2)f_1(z_1)\,dz_1=1-c_2f_2(c_2)F_1(c_1+t)=0.
$$

I now prove that $\bar{\mu}(A\backslash B)=0$.
\begin{align}
  &\bar{\mu}(A\backslash B)\nonumber\\
  &=\int_{c_2}^{c_2+b_2}\int_{c_1}^{\tilde{\phi}_1^{-1}(0)}(-z_1f_1'(z_1)f_2(z_2)-z_2f_2'(z_2)f_1(z_1)-3f_1(z_1)f_2(z_2))\,dz_1\,dz_2 \nonumber\\
  &\hspace*{.2in}-c_1f_1(c_1)\int_{c_2}^{c_2+b_2}f_2(z_2)\,dz_2+(c_2+b_2)f_2(c_2+b_2)\int_{c_1}^{\tilde{\phi}_1^{-1}(0)}f_1(z_1)\,dz_1 \nonumber\\
  &\hspace*{2.5in}-c_2f_2(c_2)\int_{c_1+t}^{\tilde{\phi}_1^{-1}(0)}f_1(z_1)\,dz_1 \nonumber\\
  &=-(\tilde{\phi}_1^{-1}(0)f_1(\tilde{\phi}_1^{-1}(0))+F_1(\tilde{\phi}_1^{-1}(0))-c_1f_1(c_1)) \nonumber\\
  &\hspace*{1in}-F_1(\tilde{\phi}_1^{-1}(0))((c_2+b_2)f_2(c_2+b_2)-c_2f_2(c_2)) \nonumber\\
  &\hspace*{.2in}-c_1f_1(c_1)+(c_2+b_2)f_2(c_2+b_2)F_1(\tilde{\phi}_1^{-1}(0))-c_2f_2(c_2)(F_1(\tilde{\phi}_1^{-1}(0))-F_1(c_1+t)) \nonumber\\
  &=-\tilde{\phi}_1^{-1}(0)f_1(\tilde{\phi}_1^{-1}(0))-F_1(\tilde{\phi}_1^{-1}(0))+c_2f_2(c_2)F_1(c_1+t)) \nonumber\\
  &=-\tilde{\phi}_1(\tilde{\phi}_1^{-1}(0))-1+c_2f_2(c_2)F_1(c_1+t) \nonumber\\
  &=0,\nonumber
\end{align}
where the final step follows because (i) $\tilde{\phi}_1(\tilde{\phi}_1^{-1}(0))=0$, and (ii) $c_2f_2(c_2)F_1(c_1+t)=1$.

I now proceed to prove that $\int_{A\backslash B}(z_1-c_1)\,d\bar{\mu}\leq0$.
\begin{align*}
  &\int_{A\backslash B}(z_1-c_1)\,d\bar{\mu} \\
  &=\int_{A\backslash B}z_1\,d\bar{\mu}-c_1\bar{\mu}(A\backslash B) \\
  &\stackrel{(a)}{=}\int_{c_2}^{c_2+b_2}\int_{c_1}^{\tilde{\phi}_1^{-1}(0)}z_1(-z_1f_1'(z_1)f_2(z_2)-z_2f_2'(z_2)f_1(z_1)-3f_1(z_1)f_2(z_2))\,dz_1\,dz_2 \\
  &\hspace*{.2in}-c_1^2f_1(c_1)\int_{c_2}^{c_2+b_2}f_2(z_2)\,dz_2+(c_2+b_2)f_2(c_2+b_2)\int_{c_1}^{\tilde{\phi}_1^{-1}(0)}z_1f_1(z_1)\,dz_1 \\
  &\hspace*{2.5in}-c_2f_2(c_2)\int_{c_1+t}^{\tilde{\phi}_1^{-1}(0)}z_1f_1(z_1)\,dz_1 \\
  &\stackrel{(b)}{=}c_1^2f_1(c_1)-(\tilde{\phi}_1^{-1}(0))^2f_1(\tilde{\phi}_1^{-1}(0)) \\
  &\hspace*{.5in}+(c_2f_2(c_2)-(c_2+b_2)f_2(c_2+b_2))\int_{c_1}^{\tilde{\phi}_1^{-1}(0)}z_1f_1(z_1)\,dz_1 \\
  &\hspace*{.75in}-c_1^2f_1(c_1)+(c_2+b_2)f_2(c_2+b_2)\int_{c_1}^{\tilde{\phi}_1^{-1}(0)}z_1f_1(z_1)\,dz_1 \\
  &\hspace*{2.5in}-c_2f_2(c_2)\int_{c_1+t}^{\tilde{\phi}_1^{-1}(0)}z_1f_1(z_1)\,dz_1 \\
  &=c_2f_2(c_2)\left(\int_{c_1}^{c_1+t}z_1f_1(z_1)\,dz_1\right)-(\tilde{\phi}_1^{-1}(0))^2f_1(\tilde{\phi}_1^{-1}(0)) \\
  &\stackrel{(c)}{\leq}0,
\end{align*}
where (a) follows since $\bar{\mu}(A\backslash B)=0$, (b) follows since
$$
  \int(-z_1^2f_1'(z_1)f_2(z_2)-2z_1f_1(z_1)f_2(z_2))\,dz_1=-z_1^2f_1(z_1)f_2(z_2),
$$
and (c) follows since $c_2f_2(c_2)\left(\int_{c_1}^{c_1+t}z_1f_1(z_1)\,dz_1\right)\leq(\tilde{\phi}_1^{-1}(0))^2f_1(\tilde{\phi}_1^{-1}(0))$.

Now I prove the claim that $\bar{\mu}(B)=0$, $\bar{\mu}(A\backslash B)=0$, and $\int_{A\backslash B}(z_1-c_1)\,d\bar{\mu}\leq 0$ is equivalent to showing $\bar{\mu}|_A\preceq_{cvx(\overrightarrow{1,-1})}0$. Let $u$ be a convex function that is nondecreasing in $z_1$ and nonincreasing in $z_2$. I now define $\tilde{u}:\left[c_1,\tilde{\phi}_1^{-1}(0)\right]\rightarrow\mathbb{R}$ to be an affine shift of the convex function obtained by truncating $u$ on the line segment $\left(\left[c_1,\tilde{\phi}_1^{-1}(0)\right)\times\{c_2+b_2\}\right)$. Specifically, I define $\tilde{u}(x)=\beta_1u(x,c_2+b_2)+\beta_2$ for some $\beta_1>0$ and $\beta_2\in\mathbb{R}$, such that $\tilde{u}(c_1)=0$ and $\tilde{u}(c_1+t)=t$. Observe that such a truncation of any convex function gives rise to a convex function because if $u$ is convex, then $u(\cdot,z_2)$ is also convex for any fixed $z_2$. Furthermore, we have $\tilde{u}(x)\leq x-c_1$ when $x\in[c_1,c_1+t]$ and $\tilde{u}(x)\geq x-c_1$ when $x\in\left[c_1+t,\tilde{\phi}_1^{-1}(0)\right)$. Now note that when $u$ is nondecreasing in $z_1$, we have
\begin{align*}
  &\int_Bu(z)\,d\bar{\mu}=u(c_1,c_2)-c_2f_2(c_2)\int_{c_1}^{c_1+t}u(z_1,c_2)f_1(z_1)\,dz_1 \\
  &\hspace*{.2in}\leq u(c_1,c_2)-u(c_1,c_2)c_2f_2(c_2)\int_{c_1}^{c_1+t}f_1(z_1)\,dz_1=u(c_1,c_2)\bar{\mu}(B)=0,
\end{align*}
where the inequality follows because $u(z_1,c_2)\geq u(c_1,c_2)$ for all $z_1\in[c_1,c_1+t]$. I now verify if $\int_{A\backslash B}u\,d\bar{\mu}\leq0$.
\begin{align*}
  &\int_{A\backslash B}u\,d\bar{\mu} \\
  &=\int_{c_2}^{c_2+b_2}u(c_1,z_2)[-c_1f_1(c_1)f_2(z_2)]\,dz_2 \\
  &\hspace*{.2in}+\int_{c_1}^{c_1+t}u(z_1,c_2+b_2)[(c_2+b_2)f_2(c_2+b_2)f_1(z_1)]\,dz_1 \\
  &\hspace*{0.5in}+\int_{c_1}^{c_1+t}\int_{c_2}^{c_2+b_2}u(z_1,z_2)[-z\cdot\nabla f(z)-3f(z)]\,dz_2\,dz_1 \\
  &\hspace*{.2in}+\int_{c_1+t}^{\tilde{\phi}_1^{-1}(0)}(u(z_1,c_2)[-c_2f_2(c_2)f_1(z_1)] \\
  &\hspace{.5in}+u(z_1,c_2+b_2)[(c_2+b_2)f_2(c_2+b_2)f_1(z_1)])\,dz_1 \\
  &\hspace*{0.5in}+\int_{c_1+t}^{\tilde{\phi}_1^{-1}(0)}\int_{c_2}^{c_2+b_2}u(z_1,z_2)[-z\cdot\nabla f(z)-3f(z)]\,dz_2\,dz_1 \\
  &\stackrel{(a)}{\leq}u(c_1,c_2+b_2)\int_{c_2}^{c_2+b_2}[-c_1f_1(c_1)f_2(z_2)]\,dz_2 \\
  &\hspace*{.2in}+\int_{c_1}^{c_1+t}u(z_1,c_2+b_2)[(c_2+b_2)f_2(c_2+b_2)f_1(z_1)]\,dz_1 \\
  &\hspace*{0.5in}+\int_{c_1}^{c_1+t}\int_{c_2}^{c_2+b_2}u(z_1,c_2+b_2)([-z\cdot\nabla f(z)-3f(z)]\,dz_2\,dz_1 \\
  &\hspace*{.2in}+\int_{c_1+t}^{\tilde{\phi}_1^{-1}(0)}u(z_1,c_2+b_2)f_1(z_1)[(c_2+b_2)f_2(c_2+b_2)-c_2f_2(c_2)]\,dz_1 \\
  &\hspace*{0.5in}+\int_{c_1+t}^{\tilde{\phi}_1^{-1}(0)}\int_{c_2}^{c_2+b_2}u(z_1,c_2+b_2)[-z\cdot\nabla f(z)-3f(z)]\,dz_2\,dz_1 \\
  &\stackrel{(b)}{=}\frac{1}{\beta_1}\tilde{u}(c_1)(-c_1f_1(c_1))+\frac{1}{\beta_1}\int_{c_1}^{c_1+t}\tilde{u}(z_1)[c_2f_2(c_2)f_1(z_1)-z_1f_1'(z_1)-2f_1(z_1)]\,dz_1 \\
  &\hspace*{1in}+\frac{1}{\beta_1}\int_{c_1+t}^{\tilde{\phi}_1^{-1}(0)}\tilde{u}(z_1)[-z_1f_1'(z_1)-2f_1(z_1)]\,dz_1-\frac{\beta_2}{\beta_1}\bar{\mu}(A\backslash B) \\
  &\stackrel{(c)}{=}\frac{1}{\beta_1}(\tilde{u}(c_1)-(c_1-c_1))(-c_1f_1(c_1)) \\
  &\hspace*{.2in}+\frac{1}{\beta_1}\int_{c_1}^{c_1+t}(\tilde{u}(z_1)-(z_1-c_1))[c_2f_2(c_2)f_1(z_1)-z_1f_1'(z_1)-2f_1(z_1)]\,dz_1 \\
  &\hspace*{.5in}+\frac{1}{\beta_1}\int_{c_1+t}^{\tilde{\phi}_1^{-1}(0)}(\tilde{u}(z_1)-(z_1-c_1))[-z_1f_1'(z_1)-2f_1(z_1)]\,dz_1 \\
  &\hspace*{2.5in}+\frac{1}{\beta_1}\int_{A\backslash B}(z_1-c_1)\,d\bar{\mu}-\frac{\beta_2}{\beta_1}\bar{\mu}(A\backslash B) \\
  &\stackrel{(d)}{\leq} 0
\end{align*}
where
\begin{enumerate}
 \item[(a)] follows since (i) $u(c_1,z_2)\geq u(c_1,c_2+b_2)$, (ii) $u(z_1,z_2)\geq u(z_1,c_2+b_2)$ for all $z_2\in[c_2,c_2+b_2]$, and (iii) $\mu(z)=-z\cdot\nabla f(z)-3f(z)\leq0$ since $f_1$ and $f_2$ are nondecreasing,
 \item[(b)] follows since $\tilde{u}(x)=\beta_1u(x,c_2+b_2)+\beta_2$,
 \item[(c)] follows by adding and subtracting $\frac{1}{\beta_1}\int_{A\backslash B}(z_1-c_1)\,d\bar{\mu}$, and
 \item[(d)] follows from (i) $\tilde{u}(c_1)=0$; (ii) $\tilde{u}(z)\leq(z-c_1)$ when $z\in[c_1,c_1+t]$ and $\tilde{u}(z)\geq(z-c_1)$ when $z\in\left[c_1+t,\tilde{\phi}_1^{-1}(0)\right]$; (iii) $c_2f_2(c_2)f_1(z_1)\geq z_1f_1'(z_1)+2f_1(z_1)$; (iv) $z_1f_1'(z_1)+2f_1(z_1)\geq0$ since $\tilde{\phi}_1$ is increasing; and (v) $\frac{1}{\beta_1}\int_{A\backslash B}(z_1-c_1)\,d\bar{\mu}\leq 0$ and $\bar{\mu}(A\backslash B)=0$.
\end{enumerate}

We finally have
$$
  \int_Au\,d\bar{\mu}=\int_Bu\,d\bar{\mu}+\int_{A\backslash B}u\,d\bar{\mu}\leq0.
$$

I have thus proved the claim that under the assumptions given in the statement of the lemma, showing $\bar{\mu}(B)=0$, $\bar{\mu}(A\backslash B)=0$ and $\int_{A\backslash B}(z_1-c_1)\,d\bar{\mu}\leq0$ is equivalent to showing $\bar{\mu}|_A\preceq_{cvx\overrightarrow{(1,-1)}}0$. \qed

{\bf Proof of Theorem \ref{thm:individual-sale}:} By Theorem \ref{thm:opt_menu}, Lemma \ref{prop:W-region} and Lemma \ref{prop:A-region}, the structure of the optimal mechanism is as in Figure \ref{fig:e1} if $f_1$ and $f_2$ positive and continuously differentiable in their domains, and in addition, if the following hold:
\begin{enumerate}
 \item[(i)] $\mu(z)\leq0$ for all $z\in D$,
 \item[(ii)] $\tilde{\phi}_1(\cdot)$ strictly increasing, $\tilde{\phi}_1(c_1)<0$. 
 \item[(iii)] $c_2f_2(c_2)f_1(z_1)\geq z_1f_1'(z_1)+2f_1(z_1),\,\forall z_1\in[c_1,\tilde{\phi}_1^{-1}(0)]$.
 \item[(iv)] There exists a $t\leq\tilde{\phi}_1^{-1}(0)-c_1$ such that $c_2f_2(c_2)F_1(c_1+t)=1$ and $c_2f_2(c_2)\left(\int_{c_1}^{c_1+t}z_1f_1(z_1)\,dz_1\right)\leq(\tilde{\phi}_1^{-1}(0))^2f_1(\tilde{\phi}_1^{-1}(0))$ hold.
\end{enumerate}
Thus part (a) of the theorem is proved if I show that the conditions (i)-(iv) are satisfied for all $f$ mentioned in the statement of the theorem, with its support set such that $c_1<\frac{1}{g_1(0)}$ and $c_2\geq c_2^*(c_1)$.

From Lemma \ref{lem:mu-phi-i}, conditions (i) and (ii) are satisfied when $f_1$ and $f_2$ are nondecreasing.

Condition (iii) is satisfied if $c_2\geq\frac{1}{g_2(0)}\left(\frac{z_1f_1'(z_1)}{f_1(z_1)}+2\right)$ for all $z_1\in[c_1,\tilde{\phi}_1^{-1}(0)]$, or if $c_2\geq\frac{1}{g_2(0)}\left(\max_{z_1\in[c_1,\tilde{\phi}_1^{-1}(0)]}\left(2+\frac{z_1f_1'(z_1)}{f_1(z_1)}\right)\right)=c_2^{th_2}(c_1)$. But $c_2^*(c_1)\geq c_2^{th_2}(c_1)$. So condition (iii) is satisfied for all $c_2\geq c_2^*(c_1)$.

Regarding condition (iv), I claim that there exists a $t\leq\tilde{\phi}_1^{-1}(0)-c_1$ such that $c_2f_2(c_2)F_1(c_1+t)=1$ whenever $c_2\geq c_2^{th_1}(c_1)$. It follows by the following series of arguments.
\begin{itemize}
 \item Let $c_2\geq c_2^{th_1}(c_1)$. Then $c_2f_2(c_2)F_1(\tilde{\phi}_1^{-1}(0))\geq 1$ holds by the definition of $c_2^{th_1}(c_1)$.
 \item Also, $c_2f_2(c_2)F_1(c_1)=0$.
 \item Now observe that $c_2f_2(c_2)F_1(c_1+t)$ is a continuous and increasing function of $t$. So if $c_2f_2(c_2)F_1(\tilde{\phi}_1^{-1}(0))\geq 1$, then there exists a $t\leq\tilde{\phi}_1^{-1}(0)-c_1$ such that $c_2f_2(c_2)F_1(c_1+t)=1$.
\end{itemize}

Now I claim that there exists a $\{t:c_2f_2(c_2)F_1(c_1+t)=1\}$ such that $c_2f_2(c_2)\left(\int_{c_1}^{c_1+t}z_1f_1(z_1)\,dz_1\right)\leq(\tilde{\phi}_1^{-1}(0))^2f_1(\tilde{\phi}_1^{-1}(0))$ whenever $c_2\geq c_2^{th_3}(c_1)$. Note that $t=t(c_2)$ as defined in Lemma \ref{prop:c2-th-3}, and that the left hand side in condition (ii) equals $\frac{\int_{c_1}^{c_1+t}z_1f_1(z_1)\,dz_1}{F_1(c_1+t)}$. Furthermore, if $z_1^*=\arg\max_{z_1\in[c_1,c_1+b_1]}(z_1(1-F_1(z_1)))$, then $z_1^*=\tilde{\phi}_1^{-1}(0)$ from Theorem \ref{thm:myerson}. Thus $\max_{z_1\in[c_1,c_1+b_1]}(z_1(1-F_1(z_1)))=\tilde{\phi}_1^{-1}(0)(1-F_1(\tilde{\phi}_1^{-1}(0)))$. But $\tilde{\phi}_1(z_1^*)=0$ implies that $z_1^*f_1(z_1^*)=(1-F_1(z_1^*))$. So we have
\begin{equation*}
  \max_{z_1\in[c_1,c_1+b_1]}(z_1(1-F_1(z_1)))=z_1^*(1-F_1(z_1^*))=(\tilde{\phi}_1^{-1}(0))^2f_1(\tilde{\phi}_1^{-1}(0)).
\end{equation*}
So condition (ii) can be rewritten as $\frac{\int_{c_1}^{c_1+t(c_2)}z_1f_1(z_1)\,dz_1}{F_1(c_1+t)}\leq\max_{z_1\in[c_1,c_1+b_1]}(z_1(1-F_1(z_1)))$, where $t(c_2)$ is as defined in Lemma \ref{prop:c2-th-3}. But the left hand side of this expression is decreasing in $c_2$ from Lemma \ref{prop:c2-th-3}. So the condition is satisfied for all $c_2\geq c_2^{th_3}(c_1)$. It is also satisfied for $c_2\geq c_2^*(c_1)$ since $c_2^*(c_1)\geq \max(c_2^{th_1}(c_1),c_2^{th_3}(c_1))$.

This proves part (a) of the theorem. The proof of part (b) is symmetric.

I now proceed to prove part (c) of the theorem. Given that $c_i\geq\frac{1}{g_i(0)}$, we have $\tilde{\phi}_i(c_i)=c_if_i(c_i)-1=c_ig_i(0)-1\geq0$. So if $\tilde{\phi}_i(c_i)\geq0$ and if $\tilde{\phi}_i(\cdot)$ is increasing, then $\tilde{\phi}_i(z_i)\geq0$ for all $z_i\in[c_i,c_i+b_i]$. Also, $\mu(z)\leq0$ for all $z\in D$. Furthermore, $f_1$ and $f_2$ are continuously differentiable on a compact interval and thus its derivatives are bounded. Part (c) of the theorem thus follows from Theorem \ref{thm:menicucci}.
\qed

{\bf Proof Theorem \ref{thm:threshold}:} I begin with the proof that $c_2^*(c_1)>\frac{1}{g_2(0)}$. Observe that
$$
  c_2^*(c_1)\geq c_2^{th_1}(c_1)=\frac{1}{g_2(0)F_1(\tilde{\phi}_1^{-1}(0))}.
$$
The proof is complete if I show that $F_1(\tilde{\phi}_1^{-1}(0))<1$. Note that $\tilde{\phi}_1^{-1}(0)<(c_1+b_1)$ from Lemma \ref{lem:mu-phi-i}(c). So $F_1(\tilde{\phi}_1^{-1}(0))<F_1(c_1+b_1)=1$, where the strict inequality holds because $F_1'(z_1)=f_1(z_1)>0$ for all $z_1$. This proves that $c_2^*(c_1)>\frac{1}{g_2(0)}$.

Now I move to proving that $c_2^*(c_1)\rightarrow\infty$ as $c_1\rightarrow\frac{1}{g_1(0)}$. Observe that the following happens as $c_1\rightarrow\frac{1}{g_1(0)}$:
 \begin{itemize}
  \item $\tilde{\phi}_1(c_1)=c_1f_1(c_1)-1=c_1g_1(0)-1$ tends to zero.
  \item So $\tilde{\phi}_1^{-1}(0)$ tends to $c_1$.
  \item Therefore $F_1(\tilde{\phi}_1^{-1}(0))\rightarrow0$.
 \end{itemize}
So $c_2^*(c_1)\geq c_2^{th_1}(c_1)=\frac{1}{g_2(0)F_1(\tilde{\phi}_1^{-1}(0))}\rightarrow\infty$ as $c_1\rightarrow\frac{1}{g_1(0)}$.

The proof for the threshold $c_1^*(c_2)$ is similar.\qed

{\bf Thresholds for Example \ref{eg:lin}:} Consider
$$
  g_1(z_1)=\frac{1}{4}(1+z_1),\,\forall z_1\in[0,2],\quad g_2(z_2)=\frac{2}{5}(2+z_2),\,\forall z_2\in[0,1].
$$
We thus have $f_1(z_1)=\frac{1}{4}(1+z_1-c_1),\,\forall z_1\in[c_1,c_1+2]$, and $f_2(z_2)=\frac{2}{5}(2+z_2-c_2),\,\forall z_2\in[c_2,c_2+1]$. Also, $F_1(z_1)=\frac{1}{8}((1+z_1-c_1)^2-1)$ and $F_2(z_2)=\frac{1}{5}((2+z_2-c_2)^2-4)$ in the relevant intervals. Now,
\begin{gather*}
  \tilde{\phi}_1(z_1)=z_1f_1(z_1)-(1-F_1(z_1))=\frac{3z_1^2+4z_1(1-c_1)+c_1^2-2c_1-8}{8}, \\
  \tilde{\phi}_2(z_2)=z_2f_2(z_2)-(1-F_2(z_2))=\frac{3z_2^2+4z_2(2-c_2)+c_2^2-4c_2-5}{5}.
\end{gather*}
Also observe that
\begin{gather*}
  \tilde{\phi}_1^{-1}(0)=\frac{2(c_1-1)+\sqrt{c_1^2-2c_1+28}}{3}, \\
  \tilde{\phi}_2^{-1}(0)=\frac{2(c_2-2)+\sqrt{c_2^2-4c_2+31}}{3}.
\end{gather*}

I now derive the thresholds $c_2^{th_k}(c_1)$, $k=1,2,3$. Observe that $c_1<\frac{1}{g_1(0)}$ when $c_1<4$. I now derive $c_2^{th_1}(c_1)$ and $c_2^{th_2}(c_1)$.
$$
  c_2^{th_1}(c_1)=\frac{1}{g_2(0)F_1(\tilde{\phi}_1^{-1}(0))}=\frac{45}{c_1^2-2c_1+10-(c_1-1)\sqrt{c_1^2-2c_1+28}}.
$$
\begin{align*}
  c_2^{th_2}(c_1)&=\frac{1}{g_2(0)}\left(\max_{z_1\in[c_1,\tilde{\phi}_1^{-1}(0)]}\left(2+\frac{z_1f_1'(z_1)}{f_1(z_1)}\right)\right)\\
  &=\frac{5}{4}\left(\max_{z_1\in[c_1,\tilde{\phi}_1^{-1}(0)]}\left(2+\frac{z_1}{z_1+(1-c_1)}\right)\right)\\
  &=\begin{cases}\frac{5}{4}\left(2+\frac{\tilde{\phi}_1^{-1}(0)}{\tilde{\phi}_1^{-1}(0)+(1-c_1)}\right)&\mbox{if }c_1\leq 1,\\\frac{5}{4}(2+c_1)&\mbox{if }c_1\in(1,4).\end{cases}
\end{align*}
I now derive $c_2^{th_3}(c_1)$. Observe that
$$
  c_2g_2(0)F_1(c_1+t)=1\Rightarrow t^2+2t-\frac{8}{c_2g_2(0)}=0\Rightarrow t(c_2)=\sqrt{1+\frac{8}{c_2g_2(0)}}-1.
$$
Furthermore, $z_1(1-F_1(z_1))$ is maximized at $z_1^*=\tilde{\phi}_1^{-1}(0)$. So I now solve for $c_2$ for which $c_2g_2(0)\left(\int_{c_1}^{c_1+\sqrt{1+\frac{8}{c_2g_2(0)}}-1}z_1f_1(z_1)\,dz_1\right)=(\tilde{\phi}_1^{-1}(0))^2f_1(\tilde{\phi}_1^{-1}(0))$.
\begin{multline*}
  \frac{c_2}{30}\left(2\left(c_1+\sqrt{1+\frac{10}{c_2}}-1\right)^3-2c_1^3+3(1-c_1)\left(\left(c_1+\sqrt{1+\frac{10}{c_2}}-1\right)^2-c_1^2\right)\right) \\
  =\frac{1}{108}\left(\left(2(c_1-1)+\sqrt{c_1^2-2c_1+28}\right)^2\left((1-c_1)+\sqrt{c_1^2-2c_1+28}\right)\right).
\end{multline*}
The value of $c_2^{th_3}(c_1)$ is computed by solving the above equation.

Now we have $c_2<\frac{1}{g_2(0)}$ when $c_2<\frac{5}{4}$. I now derive $c_1^{th_1}(c_2)$ and $c_1^{th_2}(c_2)$.
$$
  c_1^{th_1}(c_2)=\frac{1}{g_1(0)F_2(\tilde{\phi}_2^{-1}(0))}=\frac{180}{2c_2^2-8c_2-1-2(c_2-2)\sqrt{c_2^2-4c_2+31}}.
$$
\begin{align*}
  c_1^{th_2}(c_2)&=\frac{1}{g_1(0)}\left(\max_{z_2\in[c_2,\tilde{\phi}_2^{-1}(0)]}\left(2+\frac{z_2f_2'(z_2)}{f_2(z_2)}\right)\right)\\
  &=4\left(\max_{z_2\in[c_2,\tilde{\phi}_2^{-1}(0)]}\left(2+\frac{z_2}{z_2+(2-c_2)}\right)\right)\\
  &=4\left(2+\frac{\tilde{\phi}_2^{-1}(0)}{\tilde{\phi}_2^{-1}(0)+(2-c_2)}\right)
\end{align*}
I now derive $c_1^{th_3}(c_2)$. Observe that
$$
  c_1g_1(0)F_2(c_2+t)=1\Rightarrow t^2+2t-\frac{5}{c_1g_1(0)}=0\Rightarrow t(c_1)=\sqrt{4+\frac{5}{c_1g_1(0)}}-2.
$$
Furthermore, $z_2(1-F_2(z_2))$ is maximized at $z_2^*=\tilde{\phi}_2^{-1}(0)$. So I now solve for $c_1$ for which $c_1g_1(0)\left(\int_{c_2}^{c_2+\sqrt{4+\frac{5}{c_1g_1(0)}}-2}z_2f_2(z_2)\,dz_2\right)=(\tilde{\phi}_2^{-1}(0))^2f_2(\tilde{\phi}_2^{-1}(0))$.
\begin{multline*}
  \frac{c_1}{60}\left(2\left(c_2+\sqrt{4+\frac{20}{c_1}}-2\right)^3-2c_2^3+3(2-c_2)\left(\left(c_2+\sqrt{4+\frac{20}{c_1}}-2\right)^2-c_2^2\right)\right) \\
  =\frac{2}{135}\left(\left(2(c_2-2)+\sqrt{c_2^2-4c_2+31}\right)^2\left((2-c_2)+\sqrt{c_2^2-4c_2+31}\right)\right).
\end{multline*}
The value of $c_1^{th_3}(c_2)$ is computed by solving the above equation.

\section*{B. Missing Proofs in Section \ref{sec:ext}}
{\bf Proof of Theorem \ref{thm:three-high}:} Define the sets $Z$ and $W$ as follows:
\begin{itemize}
 \item $Z=\{(z_1,z_2,z_3)\in D:z_1+z_2+z_3\leq c_1+c_2+c_3+p\}$,
 \item $W=\{(z_1,z_2,z_3)\in D:z_1+z_2+z_3>c_1+c_2+c_3+p\}$,
\end{itemize}
where $p$ is chosen such that $\bar{\mu}(Z)=0$. We have $\bar{\mu}(W)=\bar{\mu}(D)-\bar{\mu}(Z)=0-0=0$. The optimal mechanism is to sell the items as a bundle if the allocation  function, $q(z)$, is given as
$$
  q(z)=\begin{cases}(0,0,0)&\mbox{if }z\in Z,\\(1,1,1)&\mbox{if }z\in W.\end{cases}
$$

So according to Theorem \ref{thm:opt_menu}, the allocation function of the optimal mechanism is as mentioned above if (i) $\bar{\mu}|_W\preceq_{\overrightarrow{(-1,-1,-1)}}0$, and (ii) $\bar{\mu}|_Z\preceq_{\overrightarrow{(1,1,1)}}0$. I now show that the measure $\bar{\mu}$ satisfies both of these conditions when $c_i\geq\frac{1}{g_i(0)}$, $i=1,2,3$.

I begin by showing that $\bar{\mu}|_Z\preceq_{\overrightarrow{(1,1,1)}}0$. Note that this is equivalent to showing that $\int_Zu\,d\bar{\mu}\leq0$ for all convex, nondecreasing $u$. We thus have
\begin{align*}
  &\int_Zu\,d\bar{\mu} \\
  &=\int_{c_3}^{c_3+p}\int_{c_2}^{c_2+c_3+p-z_3}\int_{c_1}^{c_1+c_2+c_3+p-z_2-z_3}u(z)\mu(z)\,dz_1\,dz_2\,dz_3 \\
  &\hspace*{.1in}+\int_{c_3}^{c_3+p}\int_{c_2}^{c_2+c_3+p-z_3}u(c_1,z_2,z_3)\mu_s(c_1,z_2,z_3)\,dz_2\,dz_3 \\
  &\hspace*{.1in}+\int_{c_2}^{c_2+p}\int_{c_1}^{c_1+c_2+p-z_2}u(z_1,z_2,c_3)\mu_s(z_1,z_2,c_3)\,dz_1\,dz_2 \\
  &\hspace*{.1in}+\int_{c_1}^{c_1+p}\int_{c_3}^{c_3+c_1+p-z_1}u(z_1,c_2,z_3)\mu_s(z_1,c_2,z_3)\,dz_3\,dz_1+u(c_1,c_2,c_3)\mu_p(c_1,c_2,c_3) \\
  &\leq u(c_1,c_2,c_3)\left(\int_{c_3}^{c_3+p}\int_{c_2}^{c_2+c_3+p-z_3}\int_{c_1}^{c_1+c_2+c_3+p-z_2-z_3}\mu(z)\,dz_1\,dz_2\,dz_3\right. \\
  &\hspace*{1.1in}+\int_{c_3}^{c_3+p}\int_{c_2}^{c_2+c_3+p-z_3}\mu_s(c_1,z_2,z_3)\,dz_2\,dz_3 \\
  &\hspace*{1.1in}+\int_{c_2}^{c_2+p}\int_{c_1}^{c_1+c_2+p-z_2}\mu_s(z_1,z_2,c_3)\,dz_1\,dz_2 \\
  &\hspace*{1.1in}\left.+\int_{c_1}^{c_1+p}\int_{c_3}^{c_3+c_1+p-z_1}\mu_s(z_1,c_2,z_3)\,dz_3\,dz_1+\mu_p(c_1,c_2,c_3)\right) \\
  &=u(c_1,c_2,c_3)\bar{\mu}(Z)=0,
\end{align*}
where the first inequality holds because (i) $u(z_1,z_2,z_3)\geq u(c_1,c_2,c_3)$, (ii) $\mu(z)<0$ for all $z\in D$ from Lemma \ref{lem:mu-phi-i}(a), and (iii) $\mu_s(z)<0$ for all $z$ when either $z_1=c_1$, $z_2=c_2$ or $z_3=c_3$. I have thus shown that $\bar{\mu}|_Z\preceq_{\overrightarrow{(1,1,1)}}0$.

I now proceed to show that $\bar{\mu}|_W\preceq_{\overrightarrow{(-1,-1,-1)}}0$. This is equivalent to showing that $\int_Wu\,d\bar{\mu}\geq0$ for all concave, nondecreasing $u$. I will instead prove a stronger result: $\int_Wu\,d\bar{\mu}\geq0$ for all nondecreasing $u$, which is equivalent to proving that $\bar{\mu}|_W\succeq_10$. One can prove $\bar{\mu}|_W\succeq_10$ by proving that (i) $\bar{\mu}|_W(X)\geq 0$ for any increasing set\footnote{A set $X\subseteq W$ is an increasing set if $(x_1,x_2,x_3)\in X$ implies $(y_1,y_2,y_3)\in X$ for all $\{(y_1,y_2,y_3)\in W:y_i\geq x_i,i=1,2,3\}$.} $X\subseteq W$, and (ii) $\bar{\mu}(W)=0$ \cite[Chap.~6]{SS07}. Now observe that
$$
  \int(-z_if_i'(z_i)-2f_i(z_i))\,dz_i=-(z_if_i(z_i)+F_i(z_i)).
$$
So for some $t_1,t_2,t_3>0$ such that $t_1+t_2+t_3\geq p$, we have
\begin{align}
  &\bar{\mu}|_W(\times_{i=1}^3[c_i+t_i,c_i+b_i]) \nonumber\\
  &=\int_{c_3+t_3}^{c_3+b_3}\int_{c_2+t_2}^{c_2+b_2}\int_{c_1+t_1}^{c_1+b_1}(-z_1f_1'(z_1)f_2(z_2)f_3(z_3)-z_2f_2'(z_2)f_1(z_1)f_3(z_3) \nonumber\\
  &\hspace*{1in}-z_3f_3'(z_3)f_1(z_1)f_2(z_2)-4f_1(z_1)f_2(z_2)f_3(z_3))\,dz_1\,dz_2\,dz_3 \nonumber\\
  &\hspace*{.2in}+(c_1+b_1)f_1(c_1+b_1)\int_{c_3+t_3}^{c_3+b_3}\int_{c_2+t_2}^{c_2+b_2}f_2(z_2)f_3(z_3)\,dz_2\,dz_3 \nonumber\\
  &\hspace*{.2in}+(c_2+b_2)f_2(c_2+b_2)\int_{c_1+t_1}^{c_1+b_1}\int_{c_3+t_3}^{c_3+b_3}f_3(z_3)f_1(z_1)\,dz_3\,dz_1 \nonumber\\
  &\hspace*{.2in}+(c_3+b_3)f_2(c_3+b_3)\int_{c_2+t_2}^{c_2+b_2}\int_{c_1+t_1}^{c_1+b_1}f_1(z_1)f_2(z_2)\,dz_1\,dz_2 \nonumber\\
  &=-(1-F_2(c_2+t_2))(1-F_3(c_3+t_3))((c_1+b_1)f_1(c_1+b_1)+1-(c_1+t_1)f_1(c_1+t_1)-F_1(c_1+t_1)) \nonumber\\
  &\hspace*{.2in}-(1-F_3(c_3+t_3))(1-F_1(c_1+t_1))((c_2+b_2)f_2(c_2+b_2)-(c_2+t_2)f_2(c_2+t_2)) \nonumber\\
  &\hspace*{.2in}-(1-F_1(c_1+t_1))(1-F_2(c_2+t_2))((c_3+b_3)f_3(c_3+b_3)-(c_3+t_3)f_3(c_3+t_3)) \nonumber\\
  &\hspace*{.2in}+(c_1+b_1)f_1(c_1+b_1)(1-F_2(c_2+t_2))(1-F_3(c_3+t_3)) \nonumber\\
  &\hspace*{.2in}+(c_2+b_2)f_2(c_2+b_2)(1-F_3(c_3+t_3))(1-F_1(c_1+t_1)) \nonumber\\
  &\hspace*{.2in}+(c_3+t_3)f_3(c_3+t_3))(1-F_1(c_1+t_1))(1-F_2(c_2+t_2)) \nonumber\\
  &=\tilde{\phi}_1(c_1+t_1)(1-F_2(c_2+t_2))(1-F_3(c_3+t_3)) \nonumber\\
  &\hspace*{.2in}+(c_2+t_2)f_2(c_2+t_2)(1-F_3(c_3+t_3))(1-F_1(c_1+t_1)) \nonumber\\
  &\hspace*{.3in}+(c_3+t_3)f_3(c_3+t_3)(1-F_1(c_1+t_1))(1-F_2(c_2+t_2)) \label{eqn:mu-bar-W-three}\\
  &\geq 0,\nonumber
\end{align}
where the last step follows since (i) $c_1\geq\frac{1}{g_1(0)}$ implies that $c_1f_1(c_1)-1\geq0$ which in turn implies that $\tilde{\phi}_1(c_1)=c_1f_1(c_1)-(1-F_1(c_1))\geq0$, and (ii) $\tilde{\phi}_1(\cdot)$ increasing from Lemma \ref{lem:mu-phi-i}(b) implies that $\tilde{\phi}_1(z_1)\geq0$ for all $z_1\in[c_1,c_1+b_1]$.

Now consider $t_1,t_2,t_3>0$ but $t_1+t_2+t_3<p$. Define the set $\hat{W}$ as $\hat{W}=\times_{i=1}^3[c_i+t_i,c_i+b_i]\backslash W$. Then we have
\begin{multline*}
  \bar{\mu}|_W(\times_{i=1}^3[c_i+t_i,c_i+b_i]) \\
  =\bar{\mu}(\times_{i=1}^3[c_i+t_i,c_i+b_i])-\int\int\int_{\hat{W}}\mu(z)\,dz_1\,dz_2\,dz_3\geq0,
\end{multline*}
where the inequality follows since (i) $\bar{\mu}(\times_{i=1}^3[c_i+t_i,c_i+b_i])\geq0$ can be established by a similar series of steps used to establish $\bar{\mu}|_W(\times_{i=1}^3[c_i+t_i,c_i+b_i])\geq0$ when $t_i>0$ and $\sum_{i=1}^3t_i\geq p$, and (ii) $\mu(z)<0$ for all $z\in D$ from Lemma \ref{lem:mu-phi-i}(a).

Now consider $t_3=0$. In other words, consider sets of the form $[c_1+t_1,c_1+b_1]\times[c_2+t_2,c_2+b_2]\times[c_3,c_3+b_3]$ for some $t_1,t_2>0$. Then we have
\begin{align*}
  &\bar{\mu}|_W([c_1+t_1,c_1+b_1]\times[c_2+t_2,c_2+b_2]\times[c_3,c_3+b_3]) \\
  &=\lim_{t_3\rightarrow 0}\left(\bar{\mu}|_W(\times_{i=1}^3[c_i+t_i,c_i+b_i])\right)-c_3f_3(c_3)(1-F_1(c_1+t_1))(1-F_2(c_2+t_2)) \\
  &\stackrel{(a)}{=}\tilde{\phi}_1(c_1+t_1)(1-F_2(c_2+t_2))(1-F_3(c_3+t_3)) \\
  &\hspace*{.2in}+(c_2+t_2)f_2(c_2+t_2)(1-F_3(c_3+t_3))(1-F_1(c_1+t_1)) \\
  &\geq 0,
\end{align*}
where (a) follows from \eqref{eqn:mu-bar-W-three}. Using a similar series of steps, we can show that $\bar{\mu}|_W(\times_{i=1}^3[c_i+t_i,c_i+b_i])\geq0$ when at least one of the terms $t_1$, $t_2$, or $t_3$ equals zero.

I have thus shown that $\bar{\mu}|_W(X)\geq 0$ for any increasing rectangle $X\subseteq W$. I now extend this result for increasing sets that are not rectangular. Let $c_1+t_1=\min\{z_1:(z_1,c_2+b_2,c_3+b_3)\in X\}$, $c_2+t_2=\min\{z_2:(c_1+b_1,z_2,c_3+b_3)\in X\}$, and $c_3+t_3=\min\{z_3:(c_1+b_1,c_2+b_2,z_3)\in X\}$. Observe that $X\subseteq\times_{i=1}^3[c_i+t_i,c_i+b_i]$, since $X$ is an increasing set. Let $Y=(\times_{i=1}^3[c_i+t_i,c_i+b_i])\backslash X$. Now we have
\begin{multline*}
  \bar{\mu}|_W(X)=\bar{\mu}|_W(\times_{i=1}^3[c_i+t_i,c_i+b_i])-\int\int\int_Y\mu(z)\,dz_1\,dz_2\,dz_3 \\
  \geq\bar{\mu}|_W(\times_{i=1}^3[c_i+t_i,c_i+b_i])\geq0,
\end{multline*}
where the first inequality occurs because $\mu(z)<0$ for all $z\in D$ from Lemma \ref{lem:mu-phi-i}(a). Thus $\bar{\mu}|_W(X)\geq0$ for any increasing set $X\subseteq W$, and this implies $\bar{\mu}|_W\succeq_10$.

This proves the theorem.\qed

{\bf Proof of Theorem \ref{thm:one-high}:} Note that $D=[0,1]\times[0,1]\times[c_3,c_3+1]$ for some $c_3\geq3$. The components of $\bar{\mu}$-measure can be computed from \eqref{eqn:mu}, \eqref{eqn:mu-s}, and \eqref{eqn:mu-p} as
\begin{align*}
  &\mbox{(Volume density) }\mu(z)=-4,\,\forall z\in D, \\
  &\mbox{(Area density) }\mu_s(z_1,z_2,c_3)=-c_3,\,\forall (z_1,z_2)\in[0,1]^2, \\
  &\mbox{(Area density) }\mu_s(z_1,z_2,c_3+1)=c_3+1,\,\forall (z_1,z_2)\in[0,1]^2, \\
  &\mbox{(Area density) }\mu_s(1,z_2,z_3)=1,\,\forall (z_2,z_3)\in[0,1]\times[c_3,c_3+1], \\
  &\mbox{(Area density) }\mu_s(z_1,1,z_3)=1,\,\forall (z_1,z_3)\in[0,1]\times[c_3,c_3+1], \\
  &\mbox{(Point density) }\mu_p(c_1,c_2,c_3)=1.
\end{align*}
Define the sets $Z$, $W$, $A$, and $B$ as follows:
\begin{itemize}
 \item $Z=\{(z_1,z_2,z_3)\in D:z_1\leq\frac{2}{3},z_2\leq\frac{2}{3},z_1+z_2\leq\frac{4-\sqrt{2}}{3}\}$.
 \item $A=\{(z_1,z_2,z_3)\in D:z_1\geq\frac{2}{3},z_2\leq\frac{2-\sqrt{2}}{3}\}$.
 \item $B=\{(z_1,z_2,z_3)\in D:z_2\geq\frac{2}{3},z_1\leq\frac{2-\sqrt{2}}{3}\}$.
 \item $W=D\backslash(A\cup B\cup Z)$.
\end{itemize}
The optimal mechanism is as given in the theorem statement if $q(z)$ is given as
$$
  q(z)=\begin{cases}(0,0,1)&\mbox{if }z\in Z,\\(1,0,1)&\mbox{if }z\in A,\\(0,1,1)&\mbox{if }z\in B,\\(1,1,1)&\mbox{if }z\in W.\end{cases}
$$

So according to Theorem \ref{thm:opt_menu}, the allocation function of the optimal mechanism is as mentioned above if (i) $\bar{\mu}|_Z\preceq_{\overrightarrow{(1,1,-1)}}0$, (ii) $\bar{\mu}|_A\preceq_{\overrightarrow{(-1,1,-1)}}0$, (iii) $\bar{\mu}|_B\preceq_{\overrightarrow{(1,-1,-1)}}0$, and (iv) $\bar{\mu}|_W\preceq_{\overrightarrow{(-1,-1,-1)}}0$. I now show that the measure $\bar{\mu}$ satisfies all of these conditions.

I begin by showing that $\bar{\mu}|_Z\preceq_{\overrightarrow{(1,1,-1)}}0$. This is equivalent to showing that $\int_Zu\,d\bar{\mu}\leq0$ for all $u(z_1,z_2,z_3)$ that is (i) convex, (ii) nondecreasing in $z_1,z_2$, and (iii) nonincreasing in $z_3$. So we have
\begin{align}
  &\int_Zu\,d\bar{\mu} \nonumber\\
  &=\int_{c_3}^{c_3+1}\left(\int_{0}^{\frac{2-\sqrt{2}}{3}}\int_{0}^{\frac{2}{3}}+\int_{\frac{2-\sqrt{2}}{3}}^{\frac{2}{3}}\int_{0}^{\frac{4-\sqrt{2}}{3}-z_2}\right)u(z)\mu(z)\,dz_1\,dz_2\,dz_3 \nonumber\\
  &\hspace*{.1in}+\left(\int_{0}^{\frac{2-\sqrt{2}}{3}}\int_{0}^{\frac{2}{3}}+\int_{\frac{2-\sqrt{2}}{3}}^{\frac{2}{3}}\int_{0}^{\frac{4-\sqrt{2}}{3}-z_2}\right)u(z_1,z_2,c_3+1)\mu_s(z_1,z_2,c_3+1)\,dz_1\,dz_2 \nonumber\\
  &\hspace*{.1in}+\left(\int_{0}^{\frac{2-\sqrt{2}}{3}}\int_{0}^{\frac{2}{3}}+\int_{\frac{2-\sqrt{2}}{3}}^{\frac{2}{3}}\int_{0}^{\frac{4-\sqrt{2}}{3}-z_2}\right)u(z_1,z_2,c_3)\mu_s(z_1,z_2,c_3)\,dz_1\,dz_2 \nonumber\\
  &\hspace*{.1in}+u(c_1,c_2,c_3)\mu_p(c_1,c_2,c_3) \nonumber\\
  &=(-4)\int_{c_3}^{c_3+1}\left(\int_{0}^{\frac{2-\sqrt{2}}{3}}\int_{0}^{\frac{2}{3}}+\int_{\frac{2-\sqrt{2}}{3}}^{\frac{2}{3}}\int_{0}^{\frac{4-\sqrt{2}}{3}-z_2}\right)u(z)\,dz_1\,dz_2\,dz_3 \nonumber\\
  &\hspace*{.1in}+(c_3+1+3-3)\left(\int_{0}^{\frac{2-\sqrt{2}}{3}}\int_{0}^{\frac{2}{3}}+\int_{\frac{2-\sqrt{2}}{3}}^{\frac{2}{3}}\int_{0}^{\frac{4-\sqrt{2}}{3}-z_2}\right)u(z_1,z_2,c_3+1)\,dz_1\,dz_2 \nonumber\\
  &\hspace*{.1in}-(c_3-3+3)\left(\int_{0}^{\frac{2-\sqrt{2}}{3}}\int_{0}^{\frac{2}{3}}+\int_{\frac{2-\sqrt{2}}{3}}^{\frac{2}{3}}\int_{0}^{\frac{4-\sqrt{2}}{3}-z_2}\right)u(z_1,z_2,c_3)\,dz_1\,dz_2+u(c_1,c_2,c_3).\label{eqn:Z-main}
\end{align}
Now note that
\begin{equation}\label{eqn:Z-1}
  \left(\int_{0}^{\frac{2-\sqrt{2}}{3}}\int_{0}^{\frac{2}{3}}+\int_{\frac{2-\sqrt{2}}{3}}^{\frac{2}{3}}\int_{0}^{\frac{4-\sqrt{2}}{3}-z_2}\right)\left(-4\int_{c_3}^{c_3+1}u(z)\,dz_3+4u(z_1,z_2,c_3+1)\right)\,dz_1\,dz_2\leq0
\end{equation}
since $u$ nonincreasing in $z_3$ implies that $u(z_1,z_2,c_3+1)\leq u(z_1,z_2,z_3)$ for all $z_3\in[c_3,c_3+1]$. Also note that
\begin{equation}\label{eqn:Z-2}
  (c_3-3)\left(\int_{0}^{\frac{2-\sqrt{2}}{3}}\int_{0}^{\frac{2}{3}}+\int_{\frac{2-\sqrt{2}}{3}}^{\frac{2}{3}}\int_{0}^{\frac{4-\sqrt{2}}{3}-z_2}\right)(u(z_1,z_2,c_3+1)-u(z_1,z_2,c_3))\,dz_1\,dz_2\leq0
\end{equation}
since (i) $u(z_1,z_2,c_3+1)\leq u(z_1,z_2,c_3)$, and (ii) $(c_3-3)\geq0$. Finally, note that
\begin{equation}\label{eqn:Z-3}
  u(c_1,c_2,c_3)-3\left(\int_{0}^{\frac{2-\sqrt{2}}{3}}\int_{0}^{\frac{2}{3}}+\int_{\frac{2-\sqrt{2}}{3}}^{\frac{2}{3}}\int_{0}^{\frac{4-\sqrt{2}}{3}-z_2}\right)u(z_1,z_2,c_3)\,dz_1\,dz_2\leq0
\end{equation}
since (i) $u$ nondecreasing in $(z_1,z_2)$ implies that $u(z_1,z_2,z_3)\geq u(c_1,c_2,c_3)$ for all $(z_1,z_2)\in[0,1]^2$, and (ii) $\left(\int_{0}^{\frac{2-\sqrt{2}}{3}}\int_{0}^{\frac{2}{3}}+\int_{\frac{2-\sqrt{2}}{3}}^{\frac{2}{3}}\int_{0}^{\frac{4-\sqrt{2}}{3}-z_2}\right)\,dz_1\,dz_2=\frac{1}{3}$. Given that the left hand sides of \eqref{eqn:Z-1}, \eqref{eqn:Z-2}, and \eqref{eqn:Z-3} sum to the expression in \eqref{eqn:Z-main}, we have $\int_Zu\,d\bar{\mu}\leq0$ for all the relevant $u$.

I now proceed to prove that $\bar{\mu}|_A\preceq_{\overrightarrow{(-1,1,-1)}}0$. This is equivalent to showing that $\int_Au\,d\bar{\mu}\leq0$ for all $u(z_1,z_2,z_3)$ that is (i) convex, (ii) nondecreasing in $z_2$, and (iii) nonincreasing in $z_1,z_3$. Define the sets $A_1$, $A_2$, and $A_3$ as
\begin{itemize}
 \item $A_1=[\frac{2}{3},1)\times[0,\frac{2-\sqrt{2}}{3}]\times(\{c_3\}\cup[c_3+\frac{3}{4},c_3+1])$.
 \item $A_2=[\frac{3}{4},1]\times[0,\frac{2-\sqrt{2}}{3}]\times(c_3,c_3+\frac{3}{4}]$.
 \item $A_3=(\{1\}\times[0,\frac{2-\sqrt{2}}{3}]\times(\{c_3\}\cup[c_3+\frac{3}{4},c_3+1]))\cup([\frac{2}{3},\frac{3}{4}]\times[0,\frac{2-\sqrt{2}}{3}]\times(c_3,c_3+\frac{3}{4}])$.
\end{itemize}
It is easy to check that $\bar{\mu}(A_1)=\bar{\mu}(A_2)=\bar{\mu}(A_3)=0$. We now have
\begin{align*}
  &\int_{A_1}u\,d\bar{\mu}\\
  &=\int_{0}^{\frac{2-\sqrt{2}}{3}}\int_{\frac{2}{3}}^1\left(\int_{c_3+\frac{3}{4}}^{c_3+1}u(z)\mu(z)\,dz_3\right. \\
  &\hspace*{.25in}\left.+u(z_1,z_2,c_3+1)\mu_s(z_1,z_2,c_3+1)+u(z_1,z_2,c_3)\mu_s(z_1,z_2,c_3)\right)\,dz_1\,dz_2 \\
  &\leq\int_{0}^{\frac{2-\sqrt{2}}{3}}\int_{\frac{2}{3}}^1u(z_1,z_2,c_3+1)\left(\int_{c_3+\frac{3}{4}}^{c_3+1}(-4)\,dz_3+(c_3+1)-c_3\right)\,dz_1\,dz_2 \\
  &=0,
\end{align*}
where the inequality follows since $u$ nonincreasing in $z_3$ implies that $u(z_1,z_2,c_3+1)\leq u(z_1,z_2,z_3)$ for all $z_3\in[c_3,c_3+1]$. Also, we have
\begin{align*}
  &\int_{A_2}u\,d\bar{\mu}\\
  &=\int_{0}^{\frac{2-\sqrt{2}}{3}}\int_{c_3}^{c_3+\frac{3}{4}}\left(\int_{\frac{3}{4}}^1u(z)\mu(z)\,dz_1+u(1,z_2,z_3)\mu_s(1,z_2,z_3)\right)\,dz_3\,dz_2 \\
  &\leq\int_{0}^{\frac{2-\sqrt{2}}{3}}\int_{c_3}^{c_3+\frac{3}{4}}u(1,z_2,z_3)\left(\int_{\frac{3}{4}}^1(-4)\,dz_1+1\right)\,dz_3\,dz_2 \\
  &=0,
\end{align*}
where the inequality follows since $u$ nonincreasing in $z_1$ implies that $u(1,z_2,z_3)\leq u(z_1,z_2,z_3)$ for all $z_1\in[0,1]$. Finally, we have
\begin{align*}
  &\int_{A_3}u\,d\bar{\mu}\\
  &=\int_{0}^{\frac{2-\sqrt{2}}{3}}\left(\int_{c_3+\frac{3}{4}}^{c_3+1}u(1,z_2,z_3)\mu_s(1,z_2,z_3)\,dz_3+\int_{c_3}^{c_3+\frac{3}{4}}\int_{\frac{2}{3}}^{\frac{3}{4}}u(z)\mu(z)\,dz_1\,dz_3\right)\,dz_2 \\
  &\leq\int_{0}^{\frac{2-\sqrt{2}}{3}}u\left(1,z_2,c_3+\frac{3}{4}\right)\left(\int_{c_3+\frac{3}{4}}^{c_3+1}(1)\,dz_3+\int_{c_3}^{c_3+\frac{3}{4}}\int_{\frac{2}{3}}^{\frac{3}{4}}(-4)\,dz_1\,dz_3\right)\,dz_2 \\
  &=0,
\end{align*}
where the inequality follows since (i) $u$ nonincreasing in $z_1$ implies that $u(1,z_2,z_3)\leq u(z_1,z_2,z_3)$ for all $z_1\in[0,1]$, and (ii) $u$ nonincreasing in $z_3$ implies that $u(1,z_2,z_3)\leq u(1,z_2,c_3+\frac{3}{4})\leq u(1,z_2,z_3')$ for all $z_3\in[c_3+\frac{3}{4},c_3+1], z_3'\in[c_3,c_3+\frac{3}{4}]$. Now we have
$$
  \int_Au\,d\bar{\mu}=\int_{A_1}u\,d\bar{\mu}+\int_{A_2}u\,d\bar{\mu}+\int_{A_3}u\,d\bar{\mu}\leq0.
$$
for all relevant $u$.

The proof for $\bar{\mu}|_B\preceq_{\overrightarrow{(1,-1,-1)}}0$ is exactly similar to the proof of $\bar{\mu}|_A\preceq_{\overrightarrow{(-1,1,-1)}}0$, and thus I skip the proof.

I now proceed to show that $\bar{\mu}|_W\preceq_{\overrightarrow{(-1,-1,-1)}}0$. This is equivalent to showing that $\int_Wu\,d\bar{\mu}\leq0$ for all convex, nonincreasing $u$. Define the sets $W_1$, $W_2$, and $W_3$ as
\begin{itemize}
 \item $W_1=[\frac{2}{3},1]\times[\frac{2-\sqrt{2}}{3},1)\times[c_3,c_3+1]$.
 \item $W_2=[\frac{2-\sqrt{2}}{3},\frac{2}{3}]\times[\frac{2}{3},1]\times[c_3,c_3+1]$.
 \item $W_3=([\frac{2}{3},1]\times\{1\}\times[c_3,c_3+1])\cup W_{3}'$,
\end{itemize}
where the set $W_{3}'$ is defined as
\begin{multline*}
  W_{3}'=\left(\mbox{Triangle with vertices }\left(\frac{2-\sqrt{2}}{3},\frac{2}{3},c_3\right), \left(\frac{2}{3},\frac{2}{3},c_3\right), \left(\frac{2}{3},\frac{2-\sqrt{2}}{3},c_3\right)\right) \\\times[c_3,c_3+1].
\end{multline*}
It is easy to check that $\bar{\mu}(W_1)=\bar{\mu}(W_2)=\bar{\mu}(W_3)=0$. The proof of $\int_{W_1}u\,d\bar{\mu}\leq0$ is exactly similar to the proof of $\int_Au\,d\bar{\mu}\leq0$ given that the similarity of the structure of the sets $A$ and $W_1$. Also, The proof of $\int_{W_2}u\,d\bar{\mu}\leq0$ is exactly similar to the proof of $\int_Bu\,d\bar{\mu}\leq0$ given that the similarity of the structure of the sets $B$ and $W_2$. I thus skip those proofs, and proceed to prove that $\int_{W_3}u\,d\bar{\mu}\leq0$.

\begin{align}
  &\int_{W_3}u\,d\bar{\mu} \nonumber\\
  &=\int_{c_3}^{c_3+1}\left(\int_{\frac{2}{3}}^1u(z_1,1,z_3)\mu_s(z_1,1,z_3)\,dz_1+\int_{\frac{2-\sqrt{2}}{3}}^{\frac{2}{3}}\int_{\frac{4-\sqrt{2}}{3}-z_2}^{\frac{2}{3}}u(z)\mu(z)\,dz_1\,dz_2\right)\,dz_3 \nonumber\\
  &\hspace*{.1in}+\int_{\frac{2-\sqrt{2}}{3}}^{\frac{2}{3}}\int_{\frac{4-\sqrt{2}}{3}-z_2}^{\frac{2}{3}}(u(z_1,z_2,c_3+1)\mu_s(z_1,z_2,c_3+1)+u(z_1,z_2,c_3)\mu_s(z_1,z_2,c_3))\,dz_1,dz_2 \nonumber\\
  &\leq\int_{c_3}^{c_3+1}u\left(\frac{2}{3},1,z_3\right)\left(\int_{\frac{2}{3}}^1(1)\,dz_1+\int_{\frac{2-\sqrt{2}}{3}}^{\frac{2}{3}}\int_{\frac{4-\sqrt{2}}{3}-z_2}^{\frac{2}{3}}(-3)\,dz_1\,dz_2\right)\,dz_3 \nonumber\\
  &\hspace*{.1in}+\int_{c_3}^{c_3+1}\int_{\frac{2-\sqrt{2}}{3}}^{\frac{2}{3}}\int_{\frac{4-\sqrt{2}}{3}-z_2}^{\frac{2}{3}}u(z_1,z_2,z_3)(-1)\,dz_1\,dz_2,\,dz_3 \nonumber\\
  &\hspace*{.1in}+u(z_1,z_2,c_3+1)\int_{\frac{2-\sqrt{2}}{3}}^{\frac{2}{3}}\int_{\frac{4-\sqrt{2}}{3}-z_2}^{\frac{2}{3}}(c_3+1-c_3)\,dz_1,dz_2,\label{eqn:W-main}
\end{align}
where the inequality follows since (i) $u$ nonincreasing implies $u(z_1',1,z_3)\leq u(\frac{2}{3},1,z_3)\leq u(z_1,z_2,z_3)$ for all $z_1'\in[\frac{2}{3},1]$, $(z_1,z_2)\in W_3'$, and (ii) $u$ nonincreasing implies $u(z_1,z_2,c_3+1)\leq u(z_1,z_2,z_3)$ for all $z_3\in[c_3,c_3+1]$. Now note that
\begin{equation}\label{eqn:W-1}
  \int_{\frac{2}{3}}^1(1)\,dz_1+\int_{\frac{2-\sqrt{2}}{3}}^{\frac{2}{3}}\int_{\frac{4-\sqrt{2}}{3}-z_2}^{\frac{2}{3}}(-3)\,dz_1\,dz_2=0,
\end{equation}
and also that
\begin{multline}\label{eqn:W-2}
  \int_{\frac{2-\sqrt{2}}{3}}^{\frac{2}{3}}\int_{\frac{4-\sqrt{2}}{3}-z_2}^{\frac{2}{3}}\left(\int_{c_3}^{c_3+1}(-1)u(z_1,z_2,z_3)\,dz_3+u(z_1,z_2,c_3+1)(1)\right)\,dz_1\,dz_2 \\
  \leq\int_{\frac{2-\sqrt{2}}{3}}^{\frac{2}{3}}\int_{\frac{4-\sqrt{2}}{3}-z_2}^{\frac{2}{3}}u(z_1,z_2,c_3+1)\left(1-\int_{c_3}^{c_3+1}\,dz_3\right)\,dz_1\,dz_2=0,
\end{multline}
where the inequality holds since $u$ nondecreasing implies that $u(z_1,z_2,z_3)\geq u(c_1,c_2,c_3+1)$ for all $z_3\in[c_3,c_3+1]$. Given that the left hand sides of \eqref{eqn:W-1} and \eqref{eqn:W-2} sum to the expression in \eqref{eqn:W-main}, we have $\int_{W_3}u\,d\bar{\mu}\leq0$ for all convex, nonincreasing $u$. Now we have
$$
  \int_Wu\,d\bar{\mu}=\int_{W_1}u\,d\bar{\mu}+\int_{W_2}u\,d\bar{\mu}+\int_{W_3}u\,d\bar{\mu}\leq0.
$$
for all convex, nonincreasing $u$. This completes the proof of the theorem.\qed

\section*{C. Tightness of the Threshold}
I now provide an example to show that the threshold $c_2^*(c_1)$ mentioned in Theorem \ref{thm:individual-sale} is not tight. The densities in the example violate condition (iii) in Theorem \ref{thm:suff-individual}, and thus I also show using this example that the conditions in the theorem are not necessary for the individual sale mechanism to be optimal.

I consider the same densities as in Example \ref{eg:lin}. Let $f_1(z_1)=\frac{1}{4}(1+z_1-c_1)$ when $z_1\in[c_1,c_1+2]$, and $f_2(z_2)=\frac{2}{5}(2+z_2-c_2)$ when $z_2\in[c_2,c_2+1]$. When $c_1=0$, the thresholds are given by
$$
  c_2^{th_1}(0)\approx2.9428,\,c_2^{th_2}(0)\approx3.154,\,c_2^{th_3}(0)\approx2.838,
$$
and thus $c_2^*(0)\approx3.154$. So according to Theorem \ref{thm:individual-sale}, the optimal mechanism is to sell the items individually whenever $c_1=0$ and $c_2\geq3.154$. Furthermore, note that $\tilde{\phi}_1^{-1}(0)=\frac{\sqrt{28}-2}{3}>1$, and so when $c_1=0$, $c_2=3$, and $z_1=1\in[c_1,\tilde{\phi}_1^{-1}(0)]$, we have
$$
  1.2=c_2f_2(c_2)f_1(z_1)<z_1f_1'(z_1)+2f_1(z_1)=1.25.
$$
Thus the condition (iii) of Theorem \ref{thm:suff-individual} is violated when $c_1=0$, $c_2=3$, and $z_1=1\in[c_1,\tilde{\phi}_1^{-1}(0)]$. I now show that the optimal mechanism is the individual sale when $c_2\geq3$. I would have thus shown that (a) the threshold $c_2^*(c_1)=3.154$ is not tight, and (b) the conditions in Theorem \ref{thm:suff-individual} are only sufficient but not necessary.

The densities $f_1$ and $f_2$ are positive, nondecreasing, and continuously differentiable, and thus from Lemma \ref{lem:mu-phi-i}, we have (i) $\mu(z)\leq0$ for all $z\in D$, and (ii) $\tilde{\phi}_1(\cdot)$ increasing, $\tilde{\phi}_1(c_1)<0$. Furthermore, when $c_2\geq3$, we have $c_2\geq c_2^{th_1}(0)$ and $c_2\geq c_2^{th_3}(0)$. Thus the condition that there exists $t\leq\tilde{\phi}_1^{-1}(0)-1$ such that $c_2f_2(c_2)F_1(c_1+t)=1$ and $c_2f_2(c_2)\int_{c_1}^{c_1+t}z_1f_1(z_1)\,dz_1\leq(\tilde{\phi}_1^{-1}(0))^2f_1(\tilde{\phi}_1^{-1}(0))$ is also satisfied. The only condition in Lemma \ref{prop:A-region} that is not satisfied when $c_2\in[3,c_2^{th_2}(0)]$, is the condition that $c_2f_2(c_2)f_1(z_1)\geq z_1f_1'(z_1)+2f_1(z_1)$ for all $z_1\in[c_1,\tilde{\phi}_1^{-1}(0)]$. I now show that the proof of Lemma \ref{prop:A-region} goes through with minimal modifications even when $c_2\geq3$.

We now have
$$
  \frac{1}{f_2(c_2)}\left(2+\frac{z_1f_1'(z_1)}{f_1(z_1)}\right)=\frac{5}{4}\left(2+\frac{z_1}{z_1+1}\right).
$$
The term $\frac{z_1}{z_1+1}$ increases in $z_1$ when $z_1\geq0$. The minimum thus occurs at $z_1=0$, and the maximum at $z_1=\tilde{\phi}_1^{-1}(0)=\frac{\sqrt{28}-2}{3}$. So when $c_2\leq c_2^{th_2}(0)$, we have
$$
  c_2\begin{cases}\geq\frac{1}{f_2(c_2)}\left(2+\frac{z_1f_1'(z_1)}{f_1(z_1)}\right)&\mbox{if }z_1\in[0,s(c_2)],\\<\frac{1}{f_2(c_2)}\left(2+\frac{z_1f_1'(z_1)}{f_1(z_1)}\right)&\mbox{if }z_1\in[s(c_2),\tilde{\phi}_1^{-1}(0)],\end{cases}
$$
where $s(c_2)\in[0,\tilde{\phi}_1^{-1}(0)]$ is the value of $z_1$ when $c_2=\frac{1}{f_2(c_2)}\left(2+\frac{z_1f_1'(z_1)}{f_1(z_1)}\right)$.

Recall from the proof of Lemma \ref{prop:A-region} that we first show that $\bar{\mu}(B)=\bar{\mu}(A\backslash B)=0$ and $\int_{A\backslash B}(z_1-c_1)\,d\bar{\mu}\leq0$, but neither of these proofs use the fact that $c_2f_2(c_2)f_1(z_1)\geq z_1f_1'(z_1)+2f_1(z_1)$ for all $z_1\in[c_1,\tilde{\phi}_1^{-1}(0)]$. So those results hold when $c_2\in[3,c_2^{th_2}(0)]$. I now proceed to prove that $\int_Au\,d\bar{\mu}\leq0$ for all $u$ that is (i) convex, (ii) nondecreasing in $z_1$, and (iii) nonincreasing in $z_2$.

Let $u$ be a convex function that is nondecreasing in $z_1$ and nonincreasing in $z_2$. I now define $\tilde{u}:\left[0,\tilde{\phi}_1^{-1}(0)\right]\rightarrow\mathbb{R}$ to be an affine shift of the convex function obtained by truncating $u$ on the line segment $\left(\left[0,\tilde{\phi}_1^{-1}(0)\right)\times\{c_2+1\}\right)$. Specifically, I define $\tilde{u}(x)=\beta_1u(x,c_2+1)+\beta_2$ for some $\beta_1>0$ and $\beta_2\in\mathbb{R}$, such that $\tilde{u}(0)=0$ and $\tilde{u}(s(c_2))=s(c_2)$. We now have
\begin{align*}
  &\int_{A\backslash B}u\,d\bar{\mu} \\
  &\stackrel{(a)}{\leq}\frac{1}{\beta_1}(\tilde{u}(0)-(0-0))(-0f_1(0)) \\
  &\hspace*{.2in}+\frac{1}{\beta_1}\int_{0}^{s(z_2)}(\tilde{u}(z_1)-(z_1))[c_2f_2(c_2)f_1(z_1)-z_1f_1'(z_1)-2f_1(z_1)]\,dz_1 \\
  &\hspace*{.2in}+\frac{1}{\beta_1}\int_{s(z_2)}^{t}(\tilde{u}(z_1)-(z_1))[c_2f_2(c_2)f_1(z_1)-z_1f_1'(z_1)-2f_1(z_1)]\,dz_1 \\
  &\hspace*{.5in}+\frac{1}{\beta_1}\int_{t}^{\tilde{\phi}_1^{-1}(0)}(\tilde{u}(z_1)-(z_1))[-z_1f_1'(z_1)-2f_1(z_1)]\,dz_1 \\
  &\hspace*{2.5in}+\frac{1}{\beta_1}\int_{A\backslash B}(z_1-0)\,d\bar{\mu}-\frac{\beta_2}{\beta_1}\bar{\mu}(A\backslash B) \\
  &\stackrel{(b)}{\leq} 0
\end{align*}
where
\begin{enumerate}
 \item[(a)] follows from the proof of Lemma \ref{prop:A-region}, and
 \item[(b)] follows from (i) $\tilde{u}(0)=0$; (ii) $\tilde{u}(z)\leq z$ when $z\in[0,s(z_2)]$ and $\tilde{u}(z)\geq z$ when $z\in\left[s(z_2),\tilde{\phi}_1^{-1}(0)\right]$; (iii) $c_2f_2(c_2)f_1(z_1)\geq z_1f_1'(z_1)+2f_1(z_1)$ when $z\in[0,s(z_2)]$ and $c_2f_2(c_2)f_1(z_1)\leq z_1f_1'(z_1)+2f_1(z_1)$ when $z\in[s(z_2),\tilde{\phi}_1^{-1}(0)]$; (iv) $z_1f_1'(z_1)+2f_1(z_1)\geq0$ since $\tilde{\phi}_1$ is increasing; and (v) $\frac{1}{\beta_1}\int_{A\backslash B}z_1\,d\bar{\mu}\leq 0$ and $\bar{\mu}(A\backslash B)=0$.
\end{enumerate}

This proves the result. \qed

\end{document}